\documentclass{article}
\usepackage[ruled]{algorithm2e}
\usepackage[utf8]{inputenc}
\usepackage{mathtools}
\usepackage{comment}
\usepackage{amsmath}
\usepackage{algpseudocode}
\usepackage{amsfonts,amssymb,amsthm,boxedminipage,color,url,fullpage}
\usepackage{amsfonts}
\usepackage{amssymb}
\usepackage{amsthm}
\usepackage{boxedminipage}
\usepackage{color}
\usepackage{url}
\usepackage{fullpage}
\usepackage{mathtools}
\usepackage[numbers]{natbib}

\usepackage{thmtools} 
\usepackage{thm-restate}

\usepackage{enumitem}
\usepackage{tcolorbox}
\usepackage[colorlinks=true,citecolor=blue,linkcolor=blue,urlcolor=black]{hyperref}

\usepackage[labelfont=bf]{caption}
\usepackage{aliascnt,cleveref}
\usepackage{authblk}
\usepackage{accents}
\usepackage{tikz}
\usetikzlibrary{positioning,arrows,backgrounds, calc,patterns,arrows.meta}
\usepackage{pgfplots}
\pgfplotsset{width=10cm,compat=1.9}
\usepgfplotslibrary{external}
\allowdisplaybreaks

\DeclarePairedDelimiter\floor{\lfloor}{\rfloor} 

\DeclareMathOperator*{\argmin}{arg\,min}

\theoremstyle{plain}
\newtheorem{theorem}{Theorem}[section]
\newtheorem{corollary}[theorem]{Corollary}
\newtheorem{proposition}[theorem]{Proposition}
\newtheorem*{proposition*}{Proposition}
\newtheorem{lemma}[theorem]{Lemma}
\newtheorem{maintheorem}{Main Theorem}
\newtheorem{observation}[theorem]{Observation}

\theoremstyle{definition}
\newtheorem{definition}[theorem]{Definition}
{\textcolor{red}{Question}}
\declaretheorem[style=definition,qed=$\bigtriangleup$,sibling=theorem]{example}

\newtheorem*{theorem*}{Theorem}

\theoremstyle{remark}
\newtheorem*{remark}{\upshape\bfseries Remark}

\definecolor{cobalt}{rgb}{0.0, 0.28, 0.67}

\makeatletter %% Different behavior if no arguments: https://tex.stackexchange.com/questions/474424/different-behavior-if-no-arguments
\newcommand{\EF}[1]{\if\relax\detokenize\expandafter{\@firstofone#1{}}\relax EF\xspace\else EF#1\fi}
\makeatother

\newcommand{\goods}{M}
\newcommand{\agents}{N}

\newcommand{\partitions}{\Pi}
\newcommand{\alloc}{A}
\newcommand{\allocn}{\alloc = (A_1,\dots A_n)}
\newcommand{\bundle}[1]{A_{#1}}
\newcommand{\bundlei}{\bundle{i}}
\newcommand{\MMSpartition}[2]{P^{#1} = (P^{#1}_1,\dots, P^{#1}_{#2})} 
\newcommand{\MMSnamedBundle}[2]{P^{#1}_{#2}} 
\newcommand{\MMSpartitionAnonym}[1]{P = (P_1,\dots, P_{#1})} 
\newcommand{\MMSbundle}[1]{P_{#1}} 
\newcommand{\MMSagent}[1]{P^{#1}} 
\newcommand{\MMS}[3]{\mu_{#1}^{#2}({#3})}
\newcommand{\MMSi}{\MMS{i}{n}{\goods}}
\newcommand{\MMSinstance}[2]{\mu_{#1}(#2)}

\newcommand{\GenInstance}{\calI=\langle\agents,\goods,\{v_i\}_{i \in N}\rangle}

\newcommand{\matchnfill}{\mathsf{Match\mbox{-}n\mbox{-}Fill}}

\newcommand{\bagfillcopy}{\mathsf{BagFill\mbox{-}with\mbox{-}Copies}}
\newcommand{\bagfillRR}{\mathsf{BagFill\mbox{-}RoundRobin}}

\newcommand{\reals}{\mathbb{R}}

\DeclarePairedDelimiter{\ceil}{\lceil}{\rceil}

\DeclarePairedDelimiter\norm{\lVert}{\rVert}%
\makeatletter
\let\oldnorm\norm
\def\norm{\@ifstar{\oldnorm}{\oldnorm*}}
\makeatother

\newcommand{\bbE}{\mathbb{E}}
\newcommand{\calA}{\mathcal{A}}
\newcommand{\calB}{\mathcal{B}}

\newcommand{\calE}{\mathcal{E}}

\newcommand{\calI}{\mathcal{I}}

\newcommand{\calX}{\mathcal{X}}
\newcommand{\calY}{\mathcal{Y}}

\usepackage{scrextend}

\usepackage[]{color-edits}% suppress
\definecolor{darkscarlet}{rgb}{0.34, 0.01, 0.8}
\addauthor{YG}{darkscarlet}

\definecolor{crimson}{rgb}{0.86, 0.08, 0.24}
\addauthor{MF}{crimson}

\definecolor{cadmiumgreen}{rgb}{0.0, 0.42, 0.24}
\addauthor{AE}{cadmiumgreen}

\addauthor{AF}{magenta}

\addauthor{HA}{blue}

\addauthor{SR}{brown}

\definecolor{amber}{rgb}{1.0, 0.49, 0.0}
\addauthor{SB}{amber}

\definecolor{darkpastelgreen}{rgb}{0.01, 0.75, 0.24}
\addauthor{AS}{darkpastelgreen}

\newcommand\blfootnote[1]{%
  \begingroup
  \renewcommand\thefootnote{}%
  \NoHyper\footnote{#1}\endNoHyper
  \addtocounter{footnote}{-1}%
  \endgroup
}

\title{Maximin-Share Fairness via Resource Augmentation}

\author{
Hannaneh Akrami\thanks{Bonn University and Max Planck Institute for Informatics, Germany. Email: \texttt{hakrami@uni-bonn.de}}
\quad
Siddharth Barman\thanks{Indian Institute of Science, Bangalore, India. Email: \texttt{barman@iisc.ac.in}}
\quad
Alon Eden\thanks{The Hebrew University, Jerusalem, Israel. Email: {\tt alon.eden@mail.huji.ac.il}. Incumbent of the Harry \& Abe Sherman Senior Lectureship at the School of Computer Science
and Engineering at the Hebrew University.} 
\quad 
Michal Feldman\thanks{Tel Aviv University, Israel. Email: \texttt{mfeldman@tauex.tau.ac.il}}
\quad 
Amos Fiat\thanks{Tel Aviv University, Israel. Email: \texttt{fiat@tau.ac.il}}
\quad 
Yoav Gal-Tzur\thanks{Tel Aviv University, Israel. Email: \texttt{yoavgaltzur@mail.tau.ac.il}}
\quad
Satyanand Rammohan\thanks{Technical University of Munich, Germany. Email: \texttt{satya.rammohan@tum.de}}
\quad
Aditi Sethia\thanks{Indian Institute of Science, Bangalore, India. Email: \texttt{aditisethia@iisc.ac.in}}
}

\date{\today}

\begin{document}
\maketitle
\blfootnote{This project has been partially funded by the European Research Council (ERC) under the Horizon Europe Program (grant agreement No.~101170373), by the Israel Science Foundation (grant No. 533/23), by an Amazon Research Award, by the Israel Science Foundation Breakthrough Program (grant No.~2600/24), and by a grant from TAU Center for AI and Data Science (TAD), and by the NSF-BSF (grant number 2020788). Siddharth Barman acknowledges the support of the Walmart Center for Tech Excellence (CSR WMGT-23-0001) and an Ittiam CSR Grant (OD/OTHR-24-0032).}

\begin{abstract}
We introduce and formalize the notion of resource augmentation for maximin share (MMS) fairness for the allocation of indivisible goods. 
Given an instance with $n$ agents and $m$ goods, we ask how many copies of the goods should be added in order to guarantee that each agent receives at least their original MMS value, or a meaningful approximation thereof.
\begin{itemize}
    \item For general monotone valuations, we establish a tight bound: an exact MMS allocation can be guaranteed using at most $\Theta(m/e)$ total copies, and this bound is tight even for XOS valuations. We further show that it is unavoidable to duplicate some goods $\Omega(\ln m / \ln \ln m)$ times, and provide matching upper bounds.
    \item For additive valuations, we show that at most $\min\{n-2,\floor*{\frac{m}{3}}(1+o(1))\}$ distinct copies suffice. This separates additive valuations from submodular valuations, for which we show that $n-1$ copies may be necessary.
    \item We also study approximate MMS guarantees for additive valuations and establish new tradeoffs between the number of copies needed and the approximation guaratee. In particular, we prove that $\floor{n/2}$ copies suffice to guarantee a $6/7$-approximation to the original MMS, and $\floor{n/3}$ copies suffice for a $4/5$-approximation. Both results improve upon the best-known approximation guarantees for additive valuations in the absence of copies.
    \item Finally, we relate MMS with copies to the relaxed notion of 1-out-of-$d$ MMS, showing that improvements in either framework translate directly to the other. In particular, we establish the first impossibility results for 1-out-of-$d$ MMS.
\end{itemize}
Our results highlight the power and limits of resource augmentation for achieving MMS fairness.
\end{abstract}

\thispagestyle{empty}

\newpage
\setcounter{page}{1}

\section{Introduction}\label{sec:intro} 

Fair Division dates back to antiquity, (Genesis 13), where two agents (Abraham and Lot) used cut and choose \cite{Brams_Taylor_95,Brams_Taylor_1996} to split a contested resource. 
A natural notion of fairness is \textit{proportionality}, which requires that each of the $n$ participating agents receives at least $1/n$-fraction of her value for the entire set of resources \cite{steinhaus1948problem}. Unfortunately, for indivisible goods an allocation that achieves proportionality simultaneously for all agents may not exist.

A prominent relaxation of proportionality is the \emph{maximin share} (MMS) notion \cite{Hill87,budish2011combinatorial}.
Given $n$ agents and a set $M$ of indivisible items, the maximin share of an agent is defined as the maximum value she can guarantee herself by partitioning the set of items $M$ into $n$ bundles and then receiving the least valuable bundle.
Alas, even when restricted to additive valuations, where the value of a bundle is the sum of the values of its items, an MMS allocation need not exist, even with only three agents \cite{Kurokawa18,feige2021tight}. Nonetheless, MMS allocations do exist in several natural special cases; see \cite{bouveret2016characterizing}.

\paragraph{Relaxations of MMS.}
Given that achieving the maximin share is impossible, several relaxations of MMS have been studied. Implicitly introduced in \cite{budish2011combinatorial}, a $1$-out-of-$d$ MMS is the maximum value each agent can guarantee while partitioning the set $M$ into $d > n$ bundles (instead of $n$ bundles in the original MMS definition), and receiving a minimum valued bundle.
Clearly, the larger $d$ is, the easier it is to achieve the benchmark. The question then becomes, what is the smallest value of $d$ such that there is an allocation for $n$ agents so that each agent achieves the $1$-out-of-$d$ maximin share. 
The current state of the art is that $1$-out-of-$4\lceil n/3\rceil$ MMS allocations exist~\cite{akrami2023improving}.

An $\alpha$-MMS allocation is one where every agent gets at least an $\alpha$ fraction of the maximin share. This MMS relaxation was defined by Procaccia and Wang \cite{Kurokawa18}.
They show that $\alpha=2/3$ is attainable for additive valuations. This was subsequently improved in a sequence of papers, culminating in the recent work of \cite{huang2025fptas79}, which achieved a $\frac{7}{9}$-MMS allocation, the best known bound thus far.

\paragraph{MMS via Resource Augmentation.}
In this work we take a resource augmentation approach to overcome the potential absence of MMS allocation. Namely, duplication of some of the existing goods.

The idea of using copies of goods to achieve the MMS follows naturally from the work of Budish \cite{budish2011combinatorial}, that considered a fair allocation of courses to students.
Budish introduced a mechanism with strong fairness guarantees, but with a (vanishing) ``market clearing error'', i.e., the mechanism may allocate more than the number of available seats in a class.

Resource augmentation is a central tool in computer science for beyond worst-case analysis \cite{sleator1985amortized, Roughgarden_2021}.
A prominent example is machine scheduling, where resource augmentation has been employed by increasing machine speed, adding machines, providing extra memory, and more, thereby improving performance relative to the benchmark of the non-augmented settings, or circumventing known impossibility results 
\cite{extraSpeedMahcines00,chekuri2004multi,Albers04}.

A canonical example of resource augmentation in the economic literature comes from \cite{bulowklemperer96} in the context of revenue maximization in auctions, where the augmented resources are additional bidders. They show that in a single-item auction with $n$ i.i.d.~bidders, adding just one more bidder ensures that the revenue of a simple auction exceeds the optimal revenue attainable in the original $n$-bidder market, by a complex and  arguably non-intuitive mechanism.
This result inspired a sequence of follow up works \cite{eden2017competition, liu2018competition} on what is known as ``competition complexity''.
In fact, the notion of 1-out-of-$d$ MMS is somewhat related to the competition complexity measure, where large $d$ reflects high competition, leading to lower (easier) benchmarks.

More recently, resource augmentation has been studied in problems at the intersection of computer science and economics, including
selfish routing
\cite{roughgarden2002bad,roughgarden2021beyond},
seeding in social networks  \cite{akbarpour2025just,akbarpour2018diffusion},
facility location \cite{caragiannis2016truthful},
and prophet inequalities \cite{Brustle0DEFV24,BrustleCDV24}.

\subsection{Our Contributions and Techniques}
\label{sec:full-m-over-3}

Our main contribution is the introduction of a formal framework for resource augmentation, via item duplication, for the maximin share (MMS) fairness notion in fair division. Within this framework, we present both positive and negative results for different classes of valuation functions, shedding light on the possibilities and limitations of resource augmentation for achieving MMS guarantees. Our main research question is the following:

%\begin{center}
\vspace{0.1in}
\emph{How many items must be duplicated to guarantee the existence of an allocation in which every agent receives value of at least her (original) MMS?}
\vspace{0.1in}
%\end{center}

In our analysis, we consider two measures of duplication: (i) the total number of additional copies, and (ii) the maximum number of copies of any single item. When allowing at most one extra copy per item, we refer to this setting as having {\sl distinct} copies. Our results are presented below and partly summarized in Table~\ref{tab:results}.

% \mfc{I restructure the table so that references to theorems appear in a separate line, and all text is centralized.}

\begin{table}[h]
\centering
\begin{tabular}{|c|c|c|c|}
\hline
 & \textbf{Additive} & \textbf{Submodular} & \textbf{General (Monotone)} \\
\hline
\textbf{Upper Bound} 
& 
\begin{tabular}{@{}c@{}}
$\min\{n-2,\floor*{\frac{m}{3}}(1+o(1))\}$ \\
(Thm.~\ref{thm:upperboundadditive})
\end{tabular}
& 
$m/e$ 
& 
\begin{tabular}{@{}c@{}}
$m/e$ \\
(Prop.~\ref{prop:upper_bound_t_monotone})
\end{tabular}
\\
\hline
\textbf{Lower Bound} 
& 
\begin{tabular}{@{}c@{}}
$1$ \\
\cite{Kurokawa18}
\end{tabular}
& 
\begin{tabular}{@{}c@{}}
$n-1$ \\
(Prop.~\ref{prop:submod-lb})
\end{tabular}
& 
\begin{tabular}{@{}c@{}}
$m/e$ \text{(even for XOS)} \\
(Thm.~\ref{thm:lowermontone})
\end{tabular}
\\
\hline
\end{tabular}
\caption{Bounds on the total number of copies required to achieve exact MMS. 
Results for additive valuations hold even for distinct copies.}
\label{tab:results}
\end{table}

% \begin{table}[h]
% \centering
% \begin{tabular}{lccc}
% \hline
%  & \textbf{Additive} & \textbf{Submodular} & \textbf{General (Monotone)} \\
% \hline
% \textbf{Upper Bound} & $n-2$                          & $m/e$ & $m/e$  \\
%                      & (Thm. \ref{sec:additive_goods})&       & (Prop. \ref{prop:upper_bound_t_monotone}) \\
% \textbf{Lower Bound} & $1$ & $n-1$                    & $m/e$ (XOS and beyond) \\
%                      &     & (Prop. \ref{prop:submod-lb}) & (Thm. \ref{thm:lowermontone}) \\
% \hline
% \end{tabular}
% \caption{Bounds on the total number of copies required to achieve exact MMS, with no restrictions on the number of copies per item. \ygc{TODO: add thms}}
% \label{tab:results}
% \end{table}

\paragraph{General Monotone Valuations.} 
% For general monotone valuations, 
We establish the following bound on the number of copies sufficient to guarantee an MMS allocation for general monotone valuations. Interestingly, this bound is tight, even for XOS valuations.

\begin{maintheorem}[General monotone valuations, Theorem \ref{thm:lowermontone} and Proposition~\ref{prop:upper_bound_t_monotone}]
For any set of $n$ agents with monotone valuations $\{v_i\}_{i \in [n]}$, and any set of $m$ goods, one can find an MMS allocation with copies, while making at most $\left(\frac{n-1}{n}\right)^n \cdot m$ extra copies.
Moreover, this bound is tight, even for XOS valuations.
\end{maintheorem}

Our upper bounds are established via the probabilistic method, for both the total number of copies and per-item number of copies.
Our lower bound is based on a careful construction of XOS valuations in which each agent $i$ has only $n$ designated bundles $P^i_1,\dots,P^i_n$, with the property that any bundle guaranteeing her MMS value must contain one of these $P^i_j$ as a subset. 
Moreover, the intersection between any $P^i_j$ and $P^{i'}_{j'}$ for $i \ne i'$ is large, which in turn forces the use of a large number of copies.

% Our lower bound utilizes a careful construction of XOS valuations, 
% such that each agent $i$ has only $n$ bundles $P^i_1,\dots,P^i_n$ such that any bundle that guarantees her MMS value must include one of these $P^i_j$ as a subset. Moreover, the intersection between any two $P^i_j$ and $P^{i'}_{j'}$ for $i \ne i'$ is large, suggesting the need for a large number of copies.

% such that each agent has only $n$ bundles that guarantee her MMS value, but the intersection between any two bundles satisfying different agents is large, suggesting the need for a large number of copies.

% \yge{We utilize the probabilistic method to establish upper bounds on both the total and individual number of copies required to achieve MMS. This is a new application of this technique}
% Our upper bounds for general monotone valuations rely on the probabilistic method. This approach is used both to bound the total number of copies and to bound the maximum number of copies of any single item. Thus, our work provides new applications of probabilistic techniques in the context of discrete fair division.

% Our lower bound for XOS valuations \mfc{complete...}

% \mfc{Do we want to also include this?}
In addition to bounding the total number of copies, we show that the number of copies of any individual good is at most $O(\ln m/ \ln \ln m)$, and that this bound is tight; see Theorem \ref{thm:alg_for_monotone} and Lemma \ref{lem:bound_indivcopies}.\footnote{The theorem also shows that the number of copies of an individual good can be as high as $n$ (in an instance where $m$ is sufficiently large). 
Note that this result does not conflict with the above-mentioned lower bound of $n$, since the lower bound instance has $m=n^n$.}
Theorem \ref{thm:alg_for_monotone} further asserts that, given MMS partitions for the agents, one can efficiently find---with high probability---an MMS allocation that satisfies the stated bounds on the extra copies. 
% At the same time, we prove that for some valuation classes minimizing the number of extra copies is ${\rm NP}$-hard (Theorem \ref{theorem:NPHard}).

% \mfc{Add: our techniques.}
 
\paragraph{Additive Valuations.} 
For additive valuations, all of our guarantees are achieved using {\sl distinct} copies, namely, with at most one extra copy of each item. Our main result for additive valuations is as follows:

\begin{maintheorem}[Additive valuations, Theorem~\ref{thm:upperboundadditive}]
For any instance with additive valuations, at most $\min\{n-2,\floor{m/3}(1+o(1))\}$ distinct copies suffice to guarantee an exact MMS allocation.
\end{maintheorem}

We achieve our upper bound by combining a novel bag-filling procedure which introduces additional copies together with a preliminary matching stage. Technically, this involves considering a set of minimal size which violates Hall's condition. This yields a tight bound for 3 additive agents (namely, a single copy of a single good is required to achieve MMS). 
Importantly, this bound is not tight for 4 or more agents, raising a very perplexing open problem. As far as we know today, a single copy of a single good may suffice to achieve an exact MMS allocation for any instance with additive valuations.

% Notably, this result is tight for 3 additive agents, (namely, a single copy of a single good is required to achieve MMS), but are {\sl not} known to be tight for 4 or more agents, raising a very perplexing open problem. As far as we know today, a single copy of a single good may suffice to achieve an exact MMS allocation for any instance with additive valuations.

% \mfc{Add: our techniques for additive. Some text from last version: We introduce bag filling with copies.
% Our bag filling techniques give exact MMS allocations 
% while duplicating some of the items. See Algorithm \ref{alg:bagfill_copy}: $\bagfillcopy$. }

% Interestingly, we show that for submodular valuations, $n-1$ copies are required to achieve MMS, thereby establishing a separation between additive and submodular valuations.
% (see \Cref{prop:submod-lb}).

Interestingly, we show a separation between additive and submodular valuations, where $n-1$ copies are required to guarantee an MMS allocation.

\begin{proposition*}[Lower bound for submodular valuations, \Cref{prop:submod-lb}]
    For submodular valuations, any allocation with copies which achieves MMS requires at least $n-1$ total copies. 
    % There exists an instance with submodular valuations in which any allocation with copies that achieve MMS, requires at least $n-1$ total copies. 
\end{proposition*}

% \ygc{TODO: change main theorem to natural + remove number from above prop.}
% \mfc{fixed!}

% \mfc{Think whether this result can be further emphasized, perhaps by moving it elsewhere in the intro, or even have a theorem environment for it.}

% \mfc{Add: results for submodular valuations.}

\paragraph{\boldmath $\alpha$-MMS Approximations with Copies for Additive Valuations.}

We then turn to $\alpha$-MMS for additive valuations, where we obtain new approximation guarantees that improve upon the best-known results without copies. We identify points along a tradeoff curve between the achievable $\alpha$-MMS guarantee and the number of copies required. Our main results are stated in the following theorem.

% \mfc{Need to revise the following theorem.}

% In particular, we prove that $\floor{n/2}$ copies suffice to guarantee a $6/7$-approximation to the original MMS, and $\floor{n/3}$ copies suffice for a $4/5$-approximation.

\begin{maintheorem}[$\alpha$-MMS for additive valuations, Theorem~\ref{thm:nover2copiesMMS} and~\ref{thm:nover3copiesMMS}]
For any instance with additive valuations, $\floor{n/2}$ copies suffice to guarantee a $6/7$-approximation to the original MMS, and $\floor{n/3}$ copies suffice for a $4/5$-approximation.
\end{maintheorem}

Two common techniques in the study of MMS allocations are $\alpha$-valid reductions, and considering ordered instances. 

An $\alpha$-valid reduction modifies an instance by removing a set of agents, each with an associated bundle that achieves an $\alpha$-fraction of her MMS value, while ensuring that the MMS values of the remaining agents do not decrease.
% in general take an instance and reduce it to a new instance with less agents and items, while ensuring that the MMS of each remaining agent has not decreased. An $\alpha$-valid reduction is a valid reduction where the agents removed are allocated sets whose value is at least $\alpha$ times their MMS share, see Definition \ref{def:alphavalid}. 
Our results use novel $\alpha$-valid reductions for setting with copies, see Section \ref{sec:reductions}.

% Valid reductions in general take an instance and reduce it to a new instance with less agents and items, while ensuring that the MMS of each remaining agent has not decreased. An $\alpha$-valid reduction is a valid reduction where the agents removed are allocated sets whose value is at least $\alpha$ times their MMS share, see Definition \ref{def:alphavalid}. We give new valid reductions and $\alpha$-valid reductions that use copies, see Section \ref{sec:reductions}.

An ordered instance is one where all agents have identical ranking over the items. 
It was shown in \cite{bouveret2016characterizing} that, one may assume, without loss of generality, that instances are ordered.
Unfortunately, their argument does not carry over to settings with copies.
% considering settings  agents with the same order of valuations over the items is without loss of generality when studying MMS, their proof does not carry over to the setting with copies. 
Nevertheless, using an augmenting path argument, we prove that for a restricted class of allocations, in particular those produced by our algorithms, such a reduction applies (see \Cref{sec:ordered}). 
Consequently, our algorithms for achieving approximate MMS assume, without loss of generality, ordered instances. 
It is an intriguing open problem whether such a reduction applies in general.

Combining our new $\alpha$-valid reductions with the ordered instances assumption, our algorithm further uses bag filling and round robin approaches to achieve $\alpha$-MMS allocations with copies; see Algorithm \ref{alg:RR}.

% \paragraph{\boldmath Relation to 1-out-of-$d$ MMS.}

\paragraph{\boldmath Relation to other relaxations of MMS.}
We next relate our notion of MMS with copies to other relaxations of MMS. 
We start with a relation to the notion of 1-out-of-$d$ MMS, which holds for any class of valuations.
We show that 1-out-of-$(1+\beta)n$ MMS implies that exact MMS can be achieved with $\beta m (1+o(1))$ distinct copies (see \Cref{lem:1ood_reduction}). In the other direction, one can show that, if $n\beta$ copies are sufficient for the existence of exact MMS, then 1-out-of-$\frac{n}{1-\beta}$ MMS allocations exist. 
This is because at least $(1-\beta)n$ agents do not receive any copies, and obtain their MMS values.

An interesting consequence for additive valuations is that achieving MMS with less than $n/4$ copies improves upon the currently best-known bound for 1-out-of-$d$, so such an improvement may be non-trivial.

Moreover, the example in \Cref{thm:lowermontone} which shows the lower bound of $(\frac{n-1}{n})^n \cdot m$ on the number of extra copies, also gives the following intriguing impossibility result for the existence of $1$-out-of-$d$ MMS allocations for XOS valuations (whose proof is deferred to \Cref{appendix:ordinal-XOS}). Notably, this is the first impossibility result for 1-out-of-$d$ MMS.

Additionally, for subadditive valuations, \cite{feige2025multi} showed that an $\alpha$-MMS allocation where each good is copied at most $k$ times (so the total number of copies is at most $nk$), implies the existence of a $\frac{\alpha}{4k}$-MMS allocation with no copies. 
In particular, proving the existence of such an allocation for a constant $\alpha$ and $k$, will improve upon the best known $\Omega(1/\log \log n)$-MMS for subadditive valuations given by \cite{feige2025multi}.

% Namely, for no function $f$ of the number of agents, 1-out-of-$f(n)$ MMS can be guaranteed for XOS valuations. 

\begin{maintheorem}[Impossibility of 1-out-of-$d$ for XOS valuations]
\label{thm:XOS-ordinal-impossible}
    There exists no function $f:\mathbb{N} \rightarrow \mathbb{N}$ such that for all fair division instances $\GenInstance$ with $n$ agents and XOS valuations, a $1$-out-of-$f(n)$ MMS allocation exists.
\end{maintheorem}

This impossibility result stands in stark contrast to $\alpha$-MMS guarantees for XOS valuations. Indeed, it is shown in \cite{feige2025fair} that a constant approximation (precisely, $4/17$) exists for XOS valuations.

\subsection{Open Problems}
Our framework suggests a number of intriguing open problems. We highlight several directions below.

\vspace{-0.3em}
\begin{itemize}\setlength{\itemsep}{0.005em}

\item The preeminent open problem is whether adding just one extra copy of a single good suffices to guarantee an MMS allocation for additive valuations. More broadly, a natural direction is to further narrow the existing gaps for additive valuations.

\item We reduce the problem of MMS with \emph{distinct} copies to the $1$-out-of-$d$ problem (Lemma~\ref{lem:1ood_reduction}). Consequently, any positive $1$-out-of-$d$ result for a given valuation class would immediately imply an upper bound on the number of copies required to guarantee full MMS for that class. While we establish an impossibility result for XOS valuations (Theorem~\ref{thm:XOS-ordinal-impossible}), it remains open whether any positive or negative results exist for narrower valuation classes, such as submodular or gross-substitutes. 

\item Another natural and largely unexplored direction is the design of \emph{truthful} mechanisms that achieve MMS (or approximate MMS) allocations under resource augmentation. In particular, can one design incentive-compatible mechanisms that guarantee MMS allocations with copies?

\item In standard resource-augmentation frameworks, performance is measured relative to the \emph{original} instance. An alternative direction is to study the effect of adding copies and evaluating fairness with respect to the new market.

\item Our algorithms for additive valuations duplicate at most one copy per item and aim to minimize the total number of copies. A different approach is to allow multiple copies of a single item. Even if up to $n$ copies of one item are permitted, determining the minimum total number of copies required to guarantee MMS remains open.

\item More generally, it would be interesting to explore the resource-augmentation paradigm in the context of other fairness notions, such as EF/EFX.

\end{itemize}

\subsection{Other Related Work}

Fair division of indivisible goods is a vast area of research in recent years, a survey paper by Amanatidis et al., ``Fair Division of Indivisible Goods: Recent Progress and Open Questions'', \cite{Amanatidis_2023}, includes 12 pages of references, and the authors stress that this is not an exhaustive survey and refer to yet other surveys (e.g., \cite{moulinsurvey} with a more economic perspective) that focus on specific aspects of the problem. 

Fair division notions can be roughly partitioned into share based \cite{BabaioffF22} and envy based \cite{Brams_Taylor_95}. In share based fair division one ignores what the others get, one is happy if one gets value above some threshold (that does not depend on others, only on their cardinality), proportionality and MMS are share based. In envy based fair division one decides if one is happy or not depending on the others' allocated bundles. The basic envy based fairness notion is that being envy free (one does not seek to exchange allocations with the other). For indivisible goods, envy freeness is impossible in general and thus relaxations were proposed. These include envy freeness up to one good (EF1) \cite{Lipton2004,budish2011combinatorial} and envy freeness up to any good (EFX) \cite{Caragiannis19}. 
EF1 allocations are known to exist, the question of when do EFX allocations exist is a notoriously hard open problem. The number of citations dealing with envy free, EF1 and EFX allocations is quite daunting, see \cite{Amanatidis_2023}. In this paper we deal with share based, and more specifically MMS fairness.

Prior to Budish \cite{budish2011combinatorial}, Hill also considered the maximin share \cite{Hill87} and gave lower bounds on the maximin share as a function of the item values. 

There was a long series of papers improving upon results for $\alpha$-MMS {for additive valuations}, complexity improvements and improvements in the quality of the result \cite{Amanatidis_2017,Kurokawa18,garg2019approximating,barman2020approximation}, which led to $\alpha=3/4$  \cite{ghodsi2022fair,garg2021improved,akrami2023simplification}, and slightly beyond \cite{akrami2023breaking34barrierapproximate}. 
 Recently,  \cite{heidari2025improved1013} proved the existence of a  $\frac{10}{13}$-MMS, which was then improved to $\frac{7}{9}$-MMS by \cite{huang2025fptas79}. 
 
 $\alpha$-MMS allocations have also been studied in the settings with more general valuations such as submodular \cite{barman2020approximation,ghodsi2022fair,uziahu2023,chekuri20241,chekuri2025cycle,uziahu2023}, XOS \cite{ghodsi2022fair,SEDDIGHIN2024104049,Akrami2023}, and subadditive \cite{SEDDIGHIN2024104049,feige2025,seddighin2025beating, feige2025fair}.
 Notably, and as discussed in \cite{seddighin2025beating}, one can deduce from Lemma 6.1 in \cite{SEDDIGHIN2024104049} the existence of an allocation in which every item is copied at most $\log_{3/2} n +1$ times and guarantees $1/4$-MMS for subadditive valuations.
 Interestingly, this ``multi-allocation'', as it is termed in \cite{SEDDIGHIN2024104049,seddighin2025beating}, is then used to construct an approximate MMS allocation without copies, a process was later improved by \cite{feige2025multi}.
 This exemplifies how the resource augmentation approach is not only of independent interest, but can be used as a stepping stone to improve results in the non-augmented problem.

Besides MMS and $\alpha$-MMS, other relaxations/approximations to proportionality include proportionality up to one good (PROP1), Round Robin Share, Pessimistic proportional share \cite{Conitzerprop1}, proportional up to any good (PROPx) \cite{aziz2020polynomial} which need not exist, and  proportionality up to the maximin item (PROPm) \cite{baklanov2021achievingproportionalitymaximinitem} which are known to exist for 5 agents. Allocations that are both PROP1 and Pareto efficient exist \cite{Conitzerprop1}. 
A formal definition of shares and of feasible shares appears in \cite{BabaioffF22,BabaioffEF24}, MMS itself is not feasible as it may not exist, even for additive valuations. 
A recent share-based notion is quantile-share \cite{BabichenkoFHN24}, which is guaranteed to exist for additive valuations and, more generally, for monotone valuations (then called {\sl universally feasible}), assuming the Erd\"os Matching Conjecture holds. Recently, Feige \cite{feige2025residualmaximinshare} introduced a new feasible relaxation of MMS called \emph{Residual MMS}.
Finally, the work of \cite{gafni2023unified} studies envy-freeness 
(and relaxations) for a setting in which multiple copies of each good exist, and each agent can get at most one copy of a good.

There is a long history of considering resource augmentation in various forms in computer science and economics, see chapter on resource augmentation 
 in ``Beyond the Worst-Case Analysis of Algorithms" \cite{Roughgarden_2021}. Resource augmentation appears in Bulow Klemperer \cite{bulowklemperer96} where simple auctions with one extra bidder guarantee revenue at least as large as that of complex auctions without the extra bidder. This was further studied in \cite{HartlineR09} and since then has been used in many contexts such as auction design, for both revenue and welfare, prophet inequalities, etc. It also goes by the name of competition complexity, see \cite{EdenFFTW17a,RoughgardenTY20,babaioff2019welfaremaximization,BrustleCDV24,
Brustle0DEFV24,derakhshan2024settlingcompetitioncomplexity}.
 
\section{Model and Preliminaries}\label{sec:prelims}

In this section we present our model and preliminaries for MMS allocations of goods. 

Following standard usage, define $[k]:=\{1,2, \ldots, k\}$ for integer $k \geq 1$. An instance of a fair division problem is given by a tuple $\GenInstance$, where $\goods$ is a set of $m$ indivisible goods, and $\agents$ is a set of $n$ agents. 
Every agent $i$ has a valuation function $v_i : 2^\goods \to \reals_{\geq 0}$, associating a non-negative value to every set $S \subseteq \goods$ of goods, denoted by $v_i(S)$. For subsets $S \subseteq \goods $ and $g \in \goods$ we will write $S+g$ to denote the union $S \cup \{g\}$.
As standard, we assume that $v_i$ is monotone, i.e., for any $T\subseteq S \subseteq \goods$, $v_i(T) \le v_i(S)$, and normalized, so that $v_i(\emptyset)=0$. 

In this work we also consider instances where the valuation functions are of a specific class.
A valuation function $v_i$ is said to be {\em additive} if 
for every set of items $S\subseteq \goods$, $v_i(S) = \sum_{g \in S} v_i(\{g\})$. In this case, we abuse the notation and write $v_i(g)$ or $v_{ig}$ instead of $v_i(\{g\})$ for $g \in \goods$.

An allocation of the set $\goods$ is a tuple $\allocn$, where $\bundlei$ denotes the bundle allocated to agent $i$, such that $\bigcup_{i \in [n]} \bundlei = \goods$ (i.e., all items are allocated), and $\bundlei \cap \bundle{j} = \emptyset$ for any $i \ne j$ (i.e., no item is allocated to more than one agent).
We denote the collection of all allocations of the set $\goods$ among $n$ agents by $\partitions_n(\goods)$.

The \emph{maximin share} (MMS) \cite{budish2011combinatorial} of agent $i$  is the value that agent $i$ can guarantee by splitting the set $\goods$ into $n$ bundles, getting the worst out of them. Formally,

\begin{definition}[Maximin Share (MMS)]
For an instance $\GenInstance$ with $n$ agents, the \emph{maximin share} (MMS) of agent $i$ with valuation function $v_i$, denoted by $\mu_i^n (\goods)$, is given by 
$$
\MMSi = \max_{(A_1, \ldots, A_n) \in \partitions_n(\goods)}\ \min_{1 \leq \ell \leq n} v_i \left(A_\ell \right).
$$
When convenient, we sometimes abuse notation, and refer to this value as $\MMSinstance{i}{\calI}$. 
\end{definition}

Without loss of generality, we remove all agents whose MMS value is $0$. These agents get their full MMS with empty bundles.
As valuations can be scaled, we may assume, without loss of generality that the MMS value for every agent is one. For additive valuations, this implies that every agent values the grand bundle (all items) at $n$ or more. 

Any partition that realizes the MMS value of agent $i$, is said to be an MMS partition of agent $i$. That is, a partition $\MMSpartitionAnonym{n}$ of the goods into $n$ bundles is an MMS partition of agent $i$, if and only if $\MMSi = \min_{j} v_i(\MMSbundle{j})$. 

We next present the notion of an approximate MMS allocation.

\begin{definition}[$\alpha$-MMS]
    Given a fair division instance $\GenInstance$ and $\alpha>0$, an allocation $\alloc$ is $\alpha$-MMS if and only if for all agents $i \in \agents$, $v_i(\bundlei) \geq \alpha \cdot \MMSi$.  
\end{definition}

\paragraph{Valid reductions.}
Valid reductions allow us to remove a subset of agents and items while ensuring that the MMS values of the remaining agents do not decrease.
The following lemma is due to \cite{amanatidis2018comparing,bouveret2016characterizing}.

\begin{lemma}[\cite{amanatidis2018comparing,bouveret2016characterizing}]\label{lem:MMSmonotone}
    For any agent $i \in \agents$, any monotone, non-negative, valuation functions $v_i:2^\goods \to \reals_{\geq 0}$ and any good $g \in \goods$, 
    $
    \MMSi \le \MMS{i}{n-1}{\goods \setminus \{g\}}.
    $
\end{lemma}

By applying \Cref{lem:MMSmonotone} repeatedly, we conclude that for any $k < \min\{m,n\}$, the MMS of agent $i$ may only increase if we remove an arbitrary set of $k$ goods and an arbitrary set of $k$ agents. This is shown in the following corollary.

\begin{corollary}\label{cor:MMSmonotoneK}
    For any agent $i \in \agents$, any $k < \min\{m,n\}$, and any $S \subseteq \goods$ s.t. $|S|=k$, it holds that
    $$
    \mu_i^n(\goods) \le \mu_i^{n-k}(\goods \setminus S).
    $$
\end{corollary}

We remark that a similar notion of $\alpha$-valid reductions is useful when pursuing approximate MMS allocations. This is defined and used in Section \ref{approxMMS}.

\paragraph{MMS with Copies.} An allocation with copies is an allocation where goods may be duplicated, with different {instances of the same good} allocated to different agents. An agent may receive at most one {instance} of each good. An allocation with copies is a relaxation of an allocation where the two or more bundles may have a non-empty intersection.

\begin{definition}[$(t,k)$-allocation]
    In a fair division instance $\GenInstance$, a collection of subsets of items $\allocn$ is said to be a $(t,k)$-allocation if 
    \begin{enumerate}
        \item $\bundlei \subseteq \goods$ for each agent $i \in \agents$.
        \item For each good $g \in \goods$, the number of assigned {\it extra} copies is at most $k$. That is, for each $g \in \goods$ and (copy count) $k_g \coloneqq \max \Big\{ 0, \ \ \left| \left\{ i \in \agents \mid g \in A_i \right\}\right| - 1 \Big\}$, it holds that $k_g \leq k$.
        \item The total number of \emph{extra} copies is at most $t$, i.e., $\sum_{g \in \goods} k_g \leq t$. 
    \end{enumerate}
\end{definition}
Note that, the definition above implies that no agent is allowed to receive more than a single {instance} of each good. 
We also consider valid reductions with copies which appear in Section \ref{sec:reductions}.  Also, for a $(t,k)$-allocation $\allocn$ it holds that $\sum_{i \in \agents} |\bundlei| \le |\goods|+t$.

\paragraph{Classes of Valuation Functions.}
We consider different classes of valuation functions $v:2^{M}\to \reals_+$, defined below. 
\begin{itemize}
    \item Additive: For any $S \subseteq M$, it holds that $v(S) = \sum_{g \in S} v(\{g\})$.
    \item Submodular: For any $S,T\subseteq M$, $v(S \cup T) \le v(S) + v(T) - v(S \cap T)$.
    \item XOS: There exists a finite set of additive functions $\mathcal{L}$ such that for any $S \subseteq M$, $v(S)=\max_{\ell \in \mathcal{L}} \ell(S)$.
    \item Subadditive: For any $S,T\subseteq M$, $v(S \cup T) \le v(S) + v(T)$.
\end{itemize}
It is well known that 
$\text{Additive} \subsetneq \text{Submodular}  \subsetneq 
\text{XOS}  \subsetneq 
\text{Subadditive} $ \cite{lehmann2001combinatorial}.

\section{Exact Maximin Shares}
\label{sec:exactMMS}

% \subsection{Monotone Valuations}
\subsection{XOS Valuations and Beyond}
\label{sec:monotone_goods}

In this section we show that for any instance with general monotone valuation there exists an MMS allocation using at most $m/e$ total copies and no more than $\ln m / \ln \ln m$ copies of an individual good.
We show that these bounds are tight even when all valuations are XOS (a strict subclass of subadditive valuations).
In \Cref{subsec:xos_lower_bounds} we describe our lower bound construction, in which $m/e$ total copies and $\ln m / \ln \ln m$ copies of a single good are required for achieving an MMS allocation.
In sections \ref{subsec:monotone_total_copies_upper} and \ref{subsec:general-per-good} we prove the matching upper bounds for the total and individual number of copies, respectively.

\subsubsection{\boldmath Lower Bound of $m/e$ on the Total Number of Copies for XOS Valuations}
\label{subsec:xos_lower_bounds}

Our lower bound is based on a combinatorial construction which considers $m=n^n$ goods positioned on the $n$-dimensional cube, so each good may be identified with $n$ coordinates, $(d_1,\dots,d_n)$.
The valuations are such that the agent $i$'s unique MMS partition is such that MMS bundle $j$, $P_i^j$, contains all $n^{n-1}$ goods that satisfy $d_i=j$.
The value of a certain set of goods $S$ is determined by the maximal intersection $|S \cap P_i^j|$.
A simple counting argument shows that the only way to satisfy all agents simultaneously is to copy $(n-1)^n$ goods, which is a $1/e$-fraction of the total number of goods, as $n$ tends to infinity.

\begin{theorem} \label{thm:lowermontone}
    For any $n$, there exists an instance with $n$ XOS valuations such that, the total number of copies required to achieve MMS, is a $\left(\frac{n-1}{n}\right)^n$-fraction of the goods, and there is at least one good which must be copied $(n-1)$ times.
\end{theorem}
\begin{proof}
Consider $n$ agents and $n^n$ goods positioned on the $n$-dimensional cube, i.e., the set of goods is 
$$
\goods = \{g_{d_1,\dots,d_n} \mid d_1, \dots, d_n \in [n]\}.
$$
We design agent $i$'s valuation so that her MMS partition,
$\MMSpartition{i}{n}$, is such that  
$\MMSnamedBundle{i}{j} = \{g_{d_1,\dots,d_n} \mid d_i = j \}$ and $v_i(\MMSnamedBundle{i}{j})=1$ for any $j \in [n]$. Observe that each MMS bundle has exactly $n^{n-1}$ goods.
\newcommand{\addfunc}[2]{f_{#1}^{#2}}

The valuation of agent $i$ is defined using $n$ additive functions, $\addfunc{i}{j}$ for any $j \in [n]$, where a good has positive value if only if its $i$th coordinate is $j$ (i.e., it belongs to $j$th MMS bundle of agent $i$), namely $\addfunc{i}{j}(g_{d_1,\dots,d_n}) = \frac{1}{n^{n-1}}\cdot I[d_i=j]$.
We define $v_i(S) = \max_{j \in [n]} \addfunc{i}{j}(S)$, which makes $v_i$ an XOS function.
Another way to write $v_i$ is $v_i(S) = \frac{1}{n^{n-1}}\cdot \max_{j \in [n]} |S \cap \MMSnamedBundle{i}{j}|$, which immediately implies the following observation.
\begin{observation}\label{obs:MustHavebundle}
    For any agent $i \in [n]$, $v_i(S) = 1$ if and only if there exists $j \in [n]$ such that $\MMSnamedBundle{i}{j} \subseteq S$.
\end{observation}

Observe that for any $k$ distinct agents, $i_1, \dots, i_k \in [n]$ and for any $k$ indices $j_1, \dots, j_k \in [n]$, 
$$
|\MMSnamedBundle{i_1}{j_1} \cap \MMSnamedBundle{i_2}{j_2} \cap \dots \cap \MMSnamedBundle{i_k}{j_k}|
=
|\{g_{d_1,\dots,d_n} \mid \forall l \in [k] \;\; d_{i_l} = j_l \}|
=
n^{n-k}.
$$

If we wish to find an allocation with copies $\allocn$ such that for every $i$, it holds that $v_i(\bundlei) \ge 1$, while minimizing the number of items we copy, then  by \Cref{obs:MustHavebundle} it is without loss of generality to assume that $\bundlei = \MMSnamedBundle{i}{j_i}$ for some $j_i \in [n]$.
Thus, the total number of goods (copies included) is $\sum_{i \in [n]} |\bundlei| = \sum_{i \in [n]} |\MMSnamedBundle{i}{j_i}| = n \cdot n^{n-1} = n^n$.
The number of different goods used is, using the inclusion-exclusion principle:
\begin{eqnarray*}
    \left| \bigcup_{i\in [n]} \MMSnamedBundle{i}{j_i} \right|
    &=&
    \sum_{k=1}^n (-1)^{k+1} \binom{n}{k} n^{n-k} 
    =
    n^n-n^n-\left(\sum_{k=1}^n (-1)^{k} \binom{n}{k} n^{n-k} \right) \\
    &=&
    n^n - \left(\sum_{k=0}^n (-1)^{k} \binom{n}{k} n^{n-k}\right)
    =
    n^n - (n-1)^n,
\end{eqnarray*}

which implies that the total number of copies is $(n-1)^n$.
Additionally, in any such allocation there exists a good which is allocated $n$ times, and thus duplicated $n-1$ times, as  $g_{j_1,j_2,\dots,j_n} \in \bigcap_{l=1}^n P_{j_l}^{i_l}$, and the claim follows.
\end{proof}

\begin{remark}
    The above construction implies that under XOS valuations, one cannot hope to achieve a 1-out-of-$f(n)$ MMS, for any function $f(n)$ of the number of agents, see \Cref{thm:XOS-ordinal-impossible}.
\end{remark}

\subsubsection{Upper Bound on the Total Number of Copies for General Monotone Valuations}\label{subsec:monotone_total_copies_upper}

Below we prove a matching upper bound on the number of total copies required. We show that there always exists an allocation which assigns to every agent one of her MMS bundles and makes no more than $\left(1-\frac{1}{n}\right)^n \cdot m$ additional copies.

\begin{proposition}\label{prop:upper_bound_t_monotone}
For any set $\agents$ of $n$ agents with monotone valuations $\{v_i\}_{i \in N}$, and any set $\goods$ of $m$ goods, one can find an MMS allocation with copies, while making at most $\left(\frac{n-1}{n}\right)^n \cdot m$ extra copies.
\end{proposition}

\begin{proof}
Fix $n$ MMS partitions $\{\MMSpartition{i}{n}\}_{i \in N}$.
Let $\allocn$ be a random allocation, where $\bundlei \sim \mathrm{Unif}(\{\MMSnamedBundle{i}{1},\dots,\MMSnamedBundle{i}{n}\})$, independently for any agent $i$. That is, agent $i$ gets one of her MMS bundles chosen uniformly at random.

Observe that for every agent $i$ and good $g$, $\Pr[g \in \bundlei] = \frac{1}{n}$, and so the expected number of \textit{unique} items in allocation $\alloc$ is:
\begin{eqnarray*}
\bbE_{\alloc} \left[|\cup_{i \in [n]} \bundlei|\right] 
&=& 
\sum_{g \in [m]} \Pr[g \in \cup_{i\in [n]} \bundlei] 
=
\sum_{g \in [m]} (1-\Pr[g \notin \cup_{i\in [n]} \bundlei]) 
=
\sum_{g \in [m]} \left(1-\left(1-\frac{1}{n}\right)^n \right) \\
&=&
m \cdot \left(1-\left(1-\frac{1}{n}\right)^n\right),
\end{eqnarray*}
where the third inequality follows from independence across the agents.

The expected total number of items (including copies) is
$$
\bbE_{\alloc} \left[\sum_{i=1}^n |\bundlei| \right] 
=
\sum_{i=1}^n \bbE_{\alloc} \left[ |\bundlei| \right] 
=
n \cdot \frac{m}{n}
=
m.
$$

By linearity of expectation, the expected number of copies made for an allocation $\alloc$ is
$$
\bbE_{\alloc} \left[\sum_{i=1}^n |\bundlei| - |\cup_{i \in [n]}\bundlei|\right] 
=
m - m \cdot \left(1-\left(1-\frac{1}{n}\right)^n\right)
=
m\left(1-\frac{1}{n}\right)^n,
$$
so there must exists an allocation that assigns each agent one of her MMS bundles and requires no more than $m\cdot\left(1-\frac{1}{n}\right)^n$ extra copies.
\end{proof}

\subsubsection{\boldmath Upper Bound of $\ln m/\ln \ln m$ on the Maximum Number of Individual Copies}\label{subsec:general-per-good}

In this section, we upper bound the number of copies of any individual good, and together with the bound established in \Cref{subsec:monotone_total_copies_upper}, show how to achieve both guarantees with high probability.

\begin{theorem}\label{thm:alg_for_monotone}
    For any instance $\GenInstance$ with monotone valuations, there exists a $(t,k)$-allocation that ensures MMS for  all the agents and satisfies $t \le m/e$ along with $k \le \min\left\{n-1,\frac{3\ln m}{\ln \ln m}\right\}$.
    Moreover, given an MMS partition for each agent, for any $\beta \in \mathbb{N}$, there exists an efficient algorithm that returns such an allocation with probability at least $1-e^{-\beta}$.
\end{theorem}

Notably, the upper bounds on the total and individual number of copies established in \Cref{thm:alg_for_monotone} are tight by \Cref{thm:lowermontone}, where one good is copied $n-1 \ge \frac{\ln m}{\ln \ln m}-1$ times. 
In the proof of \Cref{thm:alg_for_monotone}, we use the following lemma. Its proof, which relies on a balls-and-bins argument, is deferred to \Cref{apx:goods_proofs}.

\begin{restatable}{lemma}{lemBoundIndiv}
 \label{lem:bound_indivcopies}
    For any fair division instance $\GenInstance$, with monotone valuations $\{v_i\}_{i \in N}$, there exists an MMS allocation where every good is copied at most $O\left(\frac{\ln m}{\ln \ln m}\right)$ times.
\end{restatable}
\begin{proof}[Proof sketch]
Consider an allocation where agent $i$'s bundle is a set from her MMS partition, picked uniformly at random.
 We view this process as throwing $n$ balls (agents) into $n$ bins (the index in an MMS partition). 
 Fix some good $g \in M$, it belongs to some MMS bundle for all the agents, possibly different indices, but we can re-index the sets and assume without loss of generality that $g$ is always in the set of index 1 in the MMS partitions of all agents.  
 The event in which $g$ is copied at least $k-1$ times for the allocation $\allocn$, is exactly the event in which the bin corresponding to index $1$ has at least $k$ balls. 
 For $k=\frac{3\ln m}{\ln \ln m}$, the probability of having a bin with more than $k$ balls is at most $\frac{1}{m}$.
\end{proof}

\begin{proof}[Proof of \Cref{thm:alg_for_monotone}]
    Fix an instance $\GenInstance$. For any agent $i$, let $\MMSpartition{i}{n}$ be her MMS partition.
    Consider the algorithm $\calA$ which does the following (at most) $3m\beta$ times:
    \begin{enumerate}
        \item Samples an allocation according to $\bundlei \sim \mathrm{Unif}(\MMSagent{i})$.
        \item Returns the allocation $\alloc$ only if the number of total copies is $t \le m/e$ and every item is copied at most $k \le \frac{3\ln m}{\ln \ln m}$.
    \end{enumerate}
    First, by the union bound, the probability that algorithm $\calA$ fails in some iteration is 
    $$
    \Pr[\calA \text{ fails in a single iteration}] 
    \le \Pr[t > m/e] + \Pr\left[k > \frac{3\ln m}{\ln \ln m}\right]
    \le 
    \Pr[t > m/e] + \frac{1}{m},
    $$
    where the last inequality follows from the proof of \Cref{lem:bound_indivcopies}. 

    Recall that in the proof of \Cref{prop:upper_bound_t_monotone}, we have established that when the allocation $\alloc$ is picked uniformly at random, as in the above procedure, the expected number of total copies is $m/e$. This allows us to bound the second term using Markov's inequality,
    $$
    \Pr[t > m/e] 
    \le 
    \Pr\left[t \ge \left(1+\frac{e}{m}\right)\frac{m}{e}\right]
    \le \frac{1}{1+\frac{e}{m}} \le \frac{m}{m+2}.
    $$
    Putting everything together we have,
    \begin{align*}
        \Pr[\calA \text{ fails in a single iteration}]
        &\le
        \frac{m}{m+2} + \frac{1}{m}
        = \frac{m^2+m+2}{m^2+2m} \le \frac{m+3/2}{m+2}\le  1- \frac{1}{3m}.
    \end{align*}
    In particular, there is a strictly positive probability that $\calA$ succeeds in some iteration, so there exists an MMS allocation with $t \le m/e$ and $k \le \frac{3\ln m}{\ln \ln m}$.
    The statement follows immediately, as
    \begin{align*}
        \Pr[\calA \text{ fails in all iterations}]
        &=
        \left(\Pr[\calA \text{ fails}]\right)^{3m\beta} 
        \le
        \left(1-\frac{1}{3m}\right)^{3m\beta} 
        \le \frac{1}{e^\beta}. \qedhere
    \end{align*}
\end{proof}

\subsection{Additive Valuations} 
\label{sec:additive_goods}
In this section, we prove the following theorem.
\begin{theorem} \label{thm:upperboundadditive}
    For any instance with additive valuations, at most $\min\{n-2,\floor{m/3}(1+o(1))\}$ distinct copies suffice to guarantee an exact MMS allocation.
\end{theorem}

The proof of Theorem~\ref{thm:upperboundadditive} is divided into two parts. In Section~\ref{sec:full-n-2}, we give an algorithm that obtains exact MMS with $n-2$ distinct copies. In Section~\ref{sec:1ood_reduction} we relate the $1$-out-of-$d$ bounds to our problem (a relation that holds for arbitrary monotone valuations, and use it to obtain exact MMS with $\floor{m/3}(1+o(1))$ distinct copies for additive valuations. 

\subsubsection{\boldmath Exact MMS with $n-2$ Copies}
\label{sec:full-n-2}
Algorithm~$\matchnfill$ operates as follows. It inspects the MMS Partition $P$ of an arbitrary agent, tries to satisfy as many agents as possible (without harming other agents) using $P$ without any copies. Then it uses a modified bag-filling algorithm, $\bagfillcopy$, that uses copies in order to satisfy the remaining agents. Recall that we scale valuations so that $\mu_i^n(M)=1$ for every agent $i$, which also implies $v_i(M)\ge n$.

\begin{algorithm}[htb]
\caption{$\matchnfill(\GenInstance)$: 
\\ \textbf{Input:} $n$ additive valuations $\{v_i\}_{i\in N}$ over items in $M$. 
\\ \textbf{Output:} Allocation $A = (A_1, \ldots, A_n)$ with at most $n-2$ distinct copies.}
\label{alg:matchnfill}
\SetAlgoLined
\DontPrintSemicolon
\LinesNumbered

\While{$\exists i\in N, j\in M$ s.t. $v_{ij}\ge 1$}{ 
    $A_i\gets \{j\}$\,
    $N\gets N\setminus\{i\}$, $M\gets M\setminus \{j\}$\;
}

Let $P=(P_1,P_2,\ldots, P_n)$ be the MMS partition of agent 1\; 

Consider the bipartite graph $G=(V=N\times P,E)$, where $E=\{(i,P_j)\ : \ v_i(P_j)\ge 1\}$\;

\If{there is a perfect matching between $N$ and $P$ in $G$}{ 
    Allocate each $i\in N$ its match in the perfect matching\;
}
\Else{
    Let $P'\subset P$ be the minimal set of bundles violating Hall's condition ($|P'|>1$ since $1$ is matched to every node in $P$)\;
    Consider some $P_j\in P'$ and let $\hat{P}=P'\setminus\{P_j\}$\;
    Find a matching between bundles in $\hat{P}$ and agents in $N$\;
    \ForEach{$i\in N$ matched to some $P_j\in \hat{P}$}{
        $A_i\gets P_j$\;
    }
    Let $\tilde{N}$ be the set of unallocated agents and $\tilde{M}=M\setminus\left(\bigcup_{i\in N\setminus\tilde{N}}A_i\right)$ the set of unallocated items\;
    Allocate agents in $\tilde{N}$ using procedure $\bagfillcopy((\tilde{N},\tilde{M},\{v_i\}_{i\in \tilde{N}}))$\;
}
\end{algorithm}

\begin{algorithm}[htb]
\caption{$\bagfillcopy(\GenInstance)$: 
\\ \textbf{Input:} $n$ additive valuations $\{v_i\}_{i\in N}$ over items in $M$. 
\\ \textbf{Output:} Allocation $A = (A_1, \ldots, A_n)$ with at most $n-1$ distinct copies.}
\label{alg:bagfill_copy}
\SetAlgoLined
\LinesNumbered
$t\gets 0$, $B_t\gets \emptyset$, $N_0\gets N$, $M_0\gets M$\;

\For{items $j\gets 1$ \KwTo $m$}{
    \If{$\exists i_t\in N_t$ such that $v_{i_t}(B_t\cup \{j\})\ge 1$}{
        $A_{i_t}\gets B_t\cup \{j\}$\;
        $t\gets t+1$, $B_t\gets \{j\}$\;
        $N_t\gets N_{t-1}\setminus\{i_{t-1}\}$, $M_t\gets M_{t-1}\setminus B_{t-1}$\;
    }
    \Else{ {$B_t\gets B_t\cup \{j\}$} }
}
\end{algorithm}

We first show that $\bagfillcopy$ satisfies all agents, assuming each agent values the remaining items high enough (proof in \Cref{apx:goods_proofs}). An example of a run of $\bagfillcopy$ appears in \Cref{sec:examples}.

\begin{restatable}{lemma}{lemBagFillCopy} \label{lem:bagfillcopy}
    Assuming that for every $i\in N$, $v_i(M)\ge |N|$ and that for every item $j$, $v_{ij}<1$, $\bagfillcopy$ outputs an allocation such that $v_i(A_i)\ge 1$ for every $i\in N$, and creates at most $|N|-1$ distinct copies.
\end{restatable}
\begin{proof}[Proof sketch.]
    A simple induction shows that every time we allocate a bundle, the value of every remaining agent $i$ for all remaining items is at least the number of remaining agents. 
    Indeed, $i$'s value for all items except the last item inserted to the bag is at most $1$, otherwise, $i$ could have been allocated in a previous iteration. Since we duplicate the last inserted item, we get that the number of agents decreased by 1, while the total value of $i$ for all remaining items decreased by at most 1. Thus, we make sure that at any point, we have enough value to satisfy all remaining agents.
\end{proof}
We now show that $n-2$ distinct copies suffice.

\begin{lemma} \label{lem:nmin2_lemma}
    For any fair division instance with additive valuations, Algorithm~$\matchnfill$ guarantees an exact MMS allocation for all the agents with at most $n-2$ distinct copies.
\end{lemma}
\begin{proof}
    First, by definition, every agent that is assigned some bundle for which they have an edge in $G$ gets their MMS value. Moreover, since agent $1$ has an edge to all bundles in $P$, then at least one agent gets satisfied, and $|\tilde{N}|\le |N|-1$.
    It remains to show that sets $\tilde{N}$ and $\tilde{M}$ satisfy the conditions of Lemma~\ref{lem:bagfillcopy}.

    $\matchnfill$ first allocates items with value more than the MMS of an agent, so for every $i\in \tilde{N}, j\in \tilde{M}$, $v_{ij} < 1$. Moreover, by Lemma~\ref{lem:MMSmonotone}, allocating a single item to a single agent is a valid reduction. Thus, after this step, we have that for every $i\in N$, $v_i(M)\ge |N|$. We show that after allocating bundles $\hat{P}$ to agents in $N$, for every unallocated agent $i\in \tilde{N}$, $v_i(\tilde{M})\ge |\tilde{N}|$.

    Recall that $P'$ is a minimal set of bundles violating Hall's condition, and let $|P'|=k$. Let $N'=\Gamma(P')$  be their neighbors in $G$. Note that $|N'|=k-1$, otherwise, one can remove an arbitrary element of $P'$ and get a smaller set violating Hall's condition. Moreover, since $\hat{P}$ does not violate Hall's condition, we know that $\Gamma(\hat{P})=N'$ and that the agents in $N'$ are exactly the agents matched by bundles from $\hat{P}$. Thus, it must be that there are no edges between agents in  $\tilde{N}= N\setminus N'$ and bundles in $\hat{P}$.  Therefore, for every $i\in \tilde{N}$, $v_i(\bigcup_{S\in \hat{P}} S)\le k-1$, which implies $$v_i(\tilde{M}) = v_i(M)-v_i\left(\bigcup_{S\in \hat{P}} S\right)\ge |N|-(k-1)=|N|-|N'|=|\tilde{N}|,$$
    as desired. It follows that all the conditions required in order to apply Lemma~\ref{lem:bagfillcopy} are fulfilled, and we can satisfy agents in $\tilde{N}$ using $|\tilde{N}|-1\le|N|-2\le n-2$ distinct copies (recall that at least one agent was already allocated {a bundle} before). Since $\bagfillcopy$ is the only place where we produce copies, this concludes the proof.
\end{proof}

Together with the example in~\cite{feige2021tight} of an instance with $n=3$ additive agents that does not admit an MMS allocation, Lemma~\ref{lem:nmin2_lemma} implies a tight bound on the number of copies for $3$ additive agents. 

\begin{corollary}
    For every instance with 3 additive agents, a single copy of a single item suffices to guarantee an exact MMS allocation. This is tight, i.e., there are instances where one copy is required.
\end{corollary}

\subsubsection{\boldmath Exact MMS with $\floor{m/3}(1+o(1))$ Copies} \label{sec:1ood_reduction}

We now get an improved bound when $m$ is not too large. This is done via a reduction from 1-out-of-$(1+\alpha)n$ bounds.

\begin{lemma} \label{lem:1ood_reduction}
    If there exists a 1-out-of-$(1+\alpha)n$ MMS allocation for every fair division instance $\GenInstance$ for some constant $\alpha\in (0,1)$ then $$\floor{\alpha m} + \ceil*{\frac{(1+\alpha)^2}{n-1-\alpha}m} = \floor{\alpha m}(1+o(1))$$ 
    distinct copies suffice to guarantee an exact MMS allocation. This holds for arbitrary monotone valuations.  
\end{lemma}
\begin{proof}
    Let $\GenInstance$ be a fair division instance for which there exists a 1-out-of-$(1+\alpha)n$ MMS. Notice that if we take a set of agents $N'$ of size $n'$ for which $(1+\alpha)n'\le n$ then we can use the guarantee stated in the lemma to produce an allocation $\{A_i\}_{i\in N'}$ without copies such that for each agent $i\in N'$, $$v_i(A_i)\ge \mu_i^{\ceil{(1+\alpha)n'}}(M)\ge \mu_i^n(M).$$
    Consider an arbitrary set $N'$ of $n'=\floor{n/(1+\alpha)}$ agents and the allocation that satisfies their MMS value, $\{A_i\}_{i\in N'}$. Consider the set $\tilde{N}=N\setminus N'$. We proceed as follows. 
    
    Consider a set 
    $$S'\in \argmin_{S\subseteq N'\ :\  |S|=|\tilde{N}|} |\cup_{i\in S} A_i|.$$
    Allocate each agent $i\in S'$ the set $A_i$. Now consider the agents in $N\setminus S'$. This is a set of size $\ceil{n/(1+\alpha)},$ and therefore there exists an allocation $\{B_i\}_{i\in N\setminus S'}$ of the set $M$ of items without copies which guarantees each agent in   $N\setminus S'$ their MMS value. Allocate each agent $i\in N\setminus S'$ the set $B_i$. Notice that this allocation satisfies the MMS of all agents, and the only copies made are the items in $\cup_{i\in S'} A_i$. We next bound the number of items in this set.

    First, notice that 
    $$|S'|=|\tilde{N}|=n-\floor{n/(1+\alpha)}=\ceil*{\frac{\alpha }{1+\alpha}n}.$$
    Since we chose the set of $|\tilde{N}|$ agents with the least number of items in  $N'$, it holds that
    \begin{eqnarray*}
        |\cup_{i\in S'} A_i|& \le & \floor*{\frac{|\tilde{N}|}{|N'|}m} =  \floor*{\frac{\ceil{\frac{\alpha }{1+\alpha}n}}{\floor*{\frac{1}{1+\alpha}n}}m}
        \le  \floor*{ \frac{\frac{\alpha }{1+\alpha}n+1}{\frac{1}{1+\alpha}n-1}m} \\
        & = & \floor*{\left(\alpha  + \frac{1+\alpha}{\frac{1}{1+\alpha}n-1}\right)m}\\
        & \le & \floor{\alpha m} + \ceil*{\frac{(1+\alpha)^2}{n-1-\alpha}m } \\
        & = & \floor{\alpha m}(1+O(1/n)) = \floor{\alpha m}(1+o(1)). 
    \end{eqnarray*}
\end{proof}

We apply the above lemma to the following state-of-the-art bound due to \cite{akrami2023improving}.

\begin{lemma}[\cite{akrami2023improving}]
    For every fair division instance with additive valuations, there exist a 1-out-of-$4\ceil{n/3}$ MMS allocation.
\end{lemma}

Thus, we get the following.
\begin{corollary} \label{cor:mover3additive}
    For any fair division instance with additive valuations, at most $\floor{\frac{m}{3}}(1+o(1))$ distinct copies suffice to guarantee an exact MMS allocation.
\end{corollary}
\begin{proof}
    Since $4\ceil{n/3}\le 4n/3+4 = n(1+ \frac{1}{3} + \frac{4}{n})$, applying Lemma~\ref{lem:1ood_reduction} with $\alpha=\frac{1}{3} + \frac{4}{n}$ implies that the number of copies sufficient in order to produce an exact MMS allocation is at most
    \begin{eqnarray*}
        \floor*{\left(\frac{1}{3} + \frac{4}{n}\right)m} + \ceil*{\frac{(1+\frac{1}{3} + \frac{4}{n})^2}{n-1-\frac{1}{3} - \frac{4}{n}}m} &\le&\floor*{\frac{m}{3}} +\ceil*{\frac{4}{n}m}  + \ceil*{\frac{(1+\frac{1}{3} + \frac{4}{n})^2}{n-\frac{4}{3} - \frac{4}{n}}m}\\
        &=& \floor*{\frac{m}{3}}(1+O(1/n)) \\
    \end{eqnarray*}
\end{proof}

\subsection{\boldmath A Lower Bound of $n-1$ Total Copies for Submodular Valuations}
In this section we prove that there exists an instance where agents have submodular valuations and any MMS allocation requires at least $n-1$ total copies.
Our construction uses ideas from the submodular instance of \cite{ghodsi2022fair}.

Consider $n$ agents and $n^n$ goods positioned on the $n$-dimensional cube, $M = \{g_{d_1,\dots,d_n} \mid d_i \in [n] \;\; \forall i \in [n]\}$.
For any agent $i$ let $P_i^j = \{g_{d_1,\dots d_n} \mid d_i = j\}$, similar to the definition in \Cref{subsec:xos_lower_bounds}. 
Agent $i$'s valuation, $v_i$, is such that the MMS partition of agent $i$ is $P_i = \{P_i^j\}_{j \in [n]}$.

\begin{equation}\label{eq:submod_vals}
v_i(S) = \begin{cases}
    |S| & |S| < n^{n-1} \\
    n^{n-1} - 1/2 & |S| = n^{n-1} \text{ and } S\ne P_i^j \text{ for all } j\\
    n^{n-1} & |S| = n^{n-1} \text{ and } S=P_i^j \text{ for some } j\\
    n^{n-1} & |S| > n^{n-1}\\
\end{cases}    
\end{equation}

\begin{observation}
    $v_i$ is submodular.
\end{observation}
    
\begin{proof}
    The marginal value of all goods is 1 for $|S| < n^{n-1}-1$.
    Take $|S|=n^{n-1}-1$.
    If $S \cup \{g\} = P_i^j$, the marginal value of $g$ is 1 and no more than $1/2$ for any super set of $S$.
    If $S \cup \{g\} \ne P_i^j$, then the marginal value of $g$ is $1/2$, and no more than $1/2$ for any super set of $S$.
\end{proof}

\begin{proposition}
\label{prop:submod-lb}
There exists an instance with submodular valuations such that any allocation with copies that achieve MMS, must have at least $n-1$ total copies.
\end{proposition}
\begin{proof}
    Let $A$ be an allocation (with copies), which achieves MMS.
    Observe that $A$ must satisfy $v_i(A_i) \ge n^{n-1}$ for any $i \in [n]$.
    If agent $i$ does not get one of her MMS bundles, then to satisfy the above, it must be that $|A_i| \ge n^{n-1}+1$.
    Additionally, by using less than $n^{n-2}$ copies, only one agent $i$ can get a bundle $P^j_i$ for some $j \in [n]$, as any two bundle $P_i^j$ and $P_{i'}^{j'}$ for $i \ne i'$ share $n^{n-2}$ goods: those with $d_i=j$ and $d_{i'}=j'$. 
    Thus, the total number of items in $A$, including copies is at least $\sum_{i=1}^n |A_i| \ge n^{n-1} + (n-1)\cdot(n^{n-1}+1) = n^n+n-1$. This concludes the proof.
\end{proof}
\newcommand{\OrderedInstance}{\calI'=(N,M',\{v'_i\}_{i \in N})}
\newcommand{\tj}{j^\dagger}

\section{Approximate Maximin Shares}\label{approxMMS}

In this section we study approximate MMS allocations while allowing duplication of goods, assuming all agents have additive valuations. 

First, in \Cref{sec:ordered}, we convert a fair division instance into an \textit{ordered} instance with $n$ agents and $m$ items, in which all agents have the same ordinal preferences. 

It has been observed, for the original MMS problem (without copies), that it is without loss of generality to consider ordered instances~\cite{bouveret2016characterizing, barman2020approximation}. 
Notably, where copies are allowed, the same transformation as in~\cite{bouveret2016characterizing, barman2020approximation} may no longer yield a valid allocation, see example in \Cref{sec:picking_equence_fails}.
To address this, we impose an additional constraint on the allocation in the ordered instance (see \Cref{def:alloc_one_copy_per_agent}), and show that under this constraint, it suffices to consider ordered instances. 

\begin{restatable}{proposition}{propOrderedSimple}[Reduction to ordered for simple allocations]\label{prop:ordered_simple}
Let $\GenInstance$ be a fair division instance and let  $\OrderedInstance$ be its ordered counterpart.
Let $A'=(A'_1,\dots,A'_n)$ be a simple allocation for $\calI'$ with $t$ distinct copies.
Then, $A'$ can be converted in poly-time to an allocation for $\calI$ with $t$ distinct copies, $A=(A_1,\dots,A_n)$, such that for every agent $i$, $v_i(A_i) \ge v'_i(A'_i)$.
\end{restatable}
{We then devise algorithms that produce allocations adhering to this constraint.}

In Section \ref{sec:reductions}, we review some of the known reduction rules and introduce new ones. Then, in Section \ref{sec:RR}, we present a simple algorithm $\bagfillRR(\calI,\alpha)$ which combines the idea of bag-filling and round robin to achieve an $\alpha$-MMS allocation with copies (see Algorithm \ref{alg:RR}). This algorithm is the building block for achieving $6/7$-MMS and $4/5$-MMS with copies in the sections that follow. 
\begin{restatable}{theorem}{thmnoverTwo}\label{thm:nover2copiesMMS}
    Given an instance $\GenInstance$ with $n^*$ agents and additive valuations, for $\alpha \leq 6/7$,
    there exists an $\alpha$-MMS allocation with at most $\floor{n^*/2}$ many distinct copies.
\end{restatable}
\begin{restatable}{theorem}{thmnOverThree}\label{thm:nover3copiesMMS}
    Given an instance $\GenInstance$ with $n^* >5$ agents and additive valuations, for $\alpha \leq 4/5$, 
    there exists an $\alpha$-MMS allocation with at most $\floor{n^*/3}$ many distinct copies.
\end{restatable}

The missing proofs of this section can be found in \Cref{sec:approx-additive_proofs}.

\subsection{Ordered Instance}\label{sec:ordered}

Given an instance $\GenInstance$, for every agent $i \in N$, let $\sigma_i:[m] \to [m]$ be a permutation that orders the goods according to $i$'s preferences, namely $\sigma_i(j)$ is the $j$th most valued good (breaking ties by index\footnote{Any consistent tie breaking rule will do.}, and thus $v_i(\sigma_i(1)) \ge v_i(\sigma_i(2)) \ge \dots \ge v_i(\sigma_i(m))$.

As done by Bouveret and Lemaître  \cite{bouveret2016characterizing}, we construct an ordered instance $\OrderedInstance$ with a new set of ``virtual goods'', $M' = \{1,\dots,m\}$, and modified valuations, $\{v'_i\}_{i\in N}$, such that $v'_i(j)=v_i(\sigma_i(j))$ for any $j \in M'$.

\begin{definition}[{Ordered Counterpart} (\cite{bouveret2016characterizing})]\label{def:ordered_inst}
Given a fair division instance $\GenInstance$, its \emph{ordered counterpart}, $\OrderedInstance$, is an instance with the same set of agents and a set of ``virtual goods'', $M'$, of the same cardinality as $M$.
For each agent $i \in [n]$, and any $j \in M'$, it holds that $v'_i(j)=v_i(\sigma_i(j))$.
\end{definition}
Observe that by definition, all agent in $\calI'$ have the same ordinal preferences, i.e., for every $i \in N$, $v'_i(1)\ge v'_i(2)\ge \dots \ge v'_i(m)$.
It was shown in \cite{bouveret2016characterizing} that to guarantee agent $i$ a value of at least $\alpha_i$ in an instance $\calI$, it suffices to find an allocation providing this guarantee in its ordered counterpart $\calI'$. This is cast in the following proposition.

\begin{proposition}[Reduction to ordered without copies \cite{bouveret2016characterizing}]\label{prop:ordered_reduction_wo_copies}
Let $\GenInstance$ be a fair division instance and let  $\OrderedInstance$ be its ordered counterpart.
Any allocation (without copies) $A'=(A'_1,\dots,A'_n)$  for $\calI'$ can be converted in poly-time to an allocation (without copies) $A=(A_1,\dots,A_n)$ for $\calI$, such that for every agent $i$, $v_i(A_i) \ge v'_i(A'_i)$.
\end{proposition}
\begin{proof}
    The transformation from the ordered allocation $A'$ to the allocation $A$ is done via a simple picking sequence procedure:
    In every iteration $j \in \{1,\dots,m\}=M'$, 
    let $i$ be the agent such that $j \in A'_i$. 
    Agent $i$ chooses the most-valuable item in $M$ (according to $v_i$), which has not yet been chosen.
    
    Since by iteration $j$ exactly $j-1$ items have been chosen, agent $i$ always chooses a good $g \in M$ such that $v_i(g) \ge v_i(\sigma_i(j))=v'_i(j)$ and thus, by additivity, $v_i(A_i)\ge v'_i(A'_i)$.
\end{proof}

We will show a similar result to \Cref{prop:ordered_reduction_wo_copies}, for allocations with $t$ distinct copies in which the top $t$ goods 
(and their copies) are allocated to $2t$ different agents.
\begin{definition}[{Simple Allocation}]\label{def:alloc_one_copy_per_agent}
    Given an ordered instance $\OrderedInstance$, an allocation with $t$ total copies $A'=(A'_1,\dots,A'_n)$ is called \emph{simple} if 
    \begin{itemize}
        \item There is at most one duplication of each good.
        \item The duplicated items are $g_1,\dots,g_t$ for $t \le n/2$. 
        \item There exists a set of $2t$ \emph{distinct} agents, $i_1,\dots,i_t,i^\dagger_1,\dots,i^\dagger_t$ such that for any $j \in [t]$, agents $i_j$ and $i^\dagger_j$ are allocated the two instances of $g_j$, i.e., $A'_{i_j} \cap A'_{i^\dagger_j} = \{g_j\}$.
    \end{itemize}
\end{definition}
Recall that agent $i$ cannot get two instances of the same good in $A'_i$.\footnote{Otherwise, and if additive valuations are extended so that two instances of the same good yield twice its value, then the same argument as in \Cref{prop:ordered_reduction_wo_copies} would have implied the correctness of \Cref{prop:ordered_simple}.}

\propOrderedSimple*

In order to prove \Cref{prop:ordered_simple}, we presented the following useful lemma.

\begin{lemma}\label{lem:important}
    Fix an agent $i \in N$ and a threshold $\tau \in \reals_{+}$. 
    Let $X,Y \subset M$ be such that 
    \begin{itemize}
        \item $|X|=|Y|=k$, $X = \{x_1, \ldots, x_k\}$, $Y = \{y_1, \ldots, y_k\}$, and
        \item for all $j \in [k]$, $v_i(y_j) \geq v_i(x_j)$.
    \end{itemize}
    If there exists a set $G \subset M \setminus X$ of $r$ distinct goods such that $v_i(X \cup \{g\}) \geq \tau$ for all $g \in G$, then there exists a set $H \subset M \setminus Y$ of $r$ distinct goods such that $v_i(Y \cup \{h\}) \geq \tau$ for all $h \in H$.
\end{lemma}
\begin{proof}
    Observe that by additivity, any good $g \in G \setminus Y$ satisfies $v_i(Y \cup \{g\}) \ge v_i(X \cup \{g\}) \ge \tau$.
    Thus, it suffices to show that for every  $g \in G \cap Y$, there exists a (distinct) corresponding good $h(g) \in M \setminus (Y \cup G)$ such that $v_i(Y \cup \{h(g)\}) \geq \tau$. If this requirement is met, the set $H = (G \setminus Y) \cup \{h(g) \mid g \in G \cap Y\}$ satisfies the desired condition in the lemma. We construct the mapping from every $g$ to $h(g)$ using an augmenting path argument in a bipartite graph (see \Cref{fig:bipartite} for an illustration).
    
    Consider the directed bipartite graph $\mathcal{G} = (\calX,\calY,\calE)$, where 
    $\calX=\{\calX_g\}_{g \in X}$ and  $\calY=\{\calY_g\}_{g \in Y}$. 
    Every vertex $\calX_g \in \calX$ (respectively, $\calY_g \in \calY$ ) corresponds to good $g \in X$ (resp., $g \in Y$).
    Note that $|\calX|=|\calY|=k$.
    When clear in the context, we abuse notation and write $\calX_j$ instead of $\calX_{x_j}$ (and similarly for $\calY$).

    For every $j \in [k]$, we add the (black) edge $(\calY_j,\calX_j)$.
    In addition, for any good $g \in X \cap Y$, we add the (red) edge $(\calX_g,\calY_g)$.
    Observe that $\calE$ is a collection of disjoint cycles of length 2, and paths\footnote{
    Assume that both $X$ and $Y$ are ordered according to $i$'s valuation: $v_i(x_1) \ge \dots \ge v_i(x_k)$ and the same for $Y$, where tie breaking is consistent among $X$ and $Y$.
    Observe that path between $\calY_{y_j}$ and $\calY_{y_{j'}}$ implies that agent $i$ prefers $y_j$ over $y_{j'}$, so a cycle of length greater than 2 is not possible.
    }. 
    Additionally, every vertex $v \in \calX \cup \calY$ satisfies $deg(v) \in \{1,2\}$, where by $deg$ we consider both in- and out-degrees.
    
    Fix $g \in G \cap Y$. Since $g \in G$, then $g \notin X$ and $deg(\calY_g) = 1$. 
    Let $P = (\calY_g,\calX_{g_1},\calY_{g_1},\dots,\calX_{g_\ell},\calY_{g_\ell},\calX_{g'})$ be the unique path originating at $\calY_g$. We claim that $h(g) = g'$ satisfies the desired requirement.
    
    Since $\calX_{g'}$ has no outgoing edges, $g' \notin Y$. Also, since $g' \in X$, then $g' \notin G$.
    Thus, $g' \in M \setminus (Y \cup G)$.
    It remains to show that $v_i(Y \cup \{g'\}) \ge \tau$.
    To this end, let $Q = \{g_1,\dots,g_\ell\} \subseteq X \cap Y$, as defined by the path $P$.
    Observe that the set of goods which correspond to $\calY \cap P$ is exactly $Q \cup \{g\}$ (see \Cref{fig:bipartite}). Similarly, the set of goods corresponding to $\calX \cap P$ is exactly $Q \cup \{g'\}$.
    Observe that the vertices in $P$ can be thought of as pairs where each vertex in $\calY \cap P$ has value $\ge$ than its succeeding vertex in $\calX \cap P$. 
    Thus, when removing these edges from the graph the property $v_i(y_j) \ge v_i(x_j)$ is maintained and thus $v_i(Y \setminus (\{Q \cup \{g\})) \ge v_i(X \setminus (\{Q \cup \{g'\}))$.
    By the last inequality and additivity of $v_i$ we get
    \begin{align*}
        v_i(Y \cup \{g'\}) 
        &=
        v_i(Y \setminus (Q \cup \{g\})) + v_i(Q\cup \{g\}) + v_i(\{g'\}) \\
        &\ge
        v_i(X \setminus (Q \cup \{g'\})) + v_i(QWe \cup \{g'\}) + v_i(\{g\}) \\
        &=
        v_i(X \cup \{g\}) \\
        & \ge \tau,
    \end{align*}
    which concludes the proof.   
\end{proof}

\begin{figure}
\centering

\begin{tikzpicture}[node distance=1.5cm]

  % Nodes on top (X)
  \foreach \i/\label in {1/1, 2/2, 3/3, 4/4, 5/5} {
    \node[circle, draw, minimum size=1cm] (x\i) at (\i*2,2) {$\calX_{{\label}}$};
  }

  % Nodes on bottom (Y)
  \foreach \i/\label in {1/1, 2/2, 3/3, 4/4, 5/5} {
    \node[circle, draw, minimum size=1cm] (y\i) at (\i*2,0) {$\calY_{{\label}}$};
  }

  % Dots on the right
  \node at (12,2) {$\cdots$};
  \node at (12,0) {$\cdots$};

  % Edges y_j -> x_j
  \foreach \i in {1,...,5} {
    \draw[->] (y\i) -- (x\i);
  }

  % Custom edges
  \draw[->, thick, red] (x1) -- (y2);
  \draw[->, thick, red] (x2) -- (y3);
  \draw[->, thick, red] (x3) -- (y4);

  % Label y5 as Y(g)
  \node[below=0.3cm of y1] {\small $g$};

  % Label x2 as X(g')
  \node[above=0.3cm of x4] {\small $g'$};

  % Highlight nodes in Q
  \begin{pgfonlayer}{background}
    % Covers x2 to x5
    \path let \p1 = (x1), \p2 = (x3) in
      node[fill=blue!10, rounded corners, minimum height=1.3cm, minimum width=\x2-\x1+1.3cm, anchor=center] at ($0.5*(\p1)+0.5*(\p2)+(0,0)$) {};
    % Covers y2 to y5
    \path let \p1 = (y2), \p2 = (y4) in
      node[fill=blue!10, rounded corners, minimum height=1.3cm, minimum width=\x2-\x1+1.3cm, anchor=center] at ($0.5*(\p1)+0.5*(\p2)+(0,0)$) {};
  \end{pgfonlayer}
  
  \node[blue, below=0.3cm of y3] {$Q$};
  \node[blue, above=0.3cm of x2] {$Q$};

\end{tikzpicture}
    \caption{Illustration of the augmenting path argument in the proof of \Cref{lem:important}. The element $g \in G$ (whose corresponding node is $\mathcal{Y}_g$, $g=y_1$) is replaced with $g'$ (whose corresponding node is $\mathcal{X}_{g'}=\mathcal{X}_4$, $g'=x_4$). We write $\mathcal{X}_i$ / $\mathcal{Y}_i$ rather than the cumbersome $\mathcal{X}_{x_i}$/$\mathcal{Y}_{y_i}$. It is purely an example that the path nodes are associated with a contiguous prefix of the goods, e.g., $\mathcal{X}_1$ could point to $\mathcal{Y}_5$.}
    \label{fig:bipartite}
\end{figure}

\begin{proof}[Proof of \Cref{prop:ordered_simple}]
Let $A' = (A'_1,\dots,A'_n)$ be a simple allocation {in which the goods $g_1< \ldots< g_t$ for $t \leq n/2$ are copied}. {Rename the agents so that for every $j \in [t]$ one copy of $g_j$ is allocated to agent $j$, i.e., $g_j \in A_j$.}
We will refer to each such {$g_j$} as the `copy', and to the other instance of {$g_j$} (assigned to some agent $i \in N \setminus \{1,\dots,t\}$), as the `original'.

We will use a similar picking sequence procedure to the one given in \Cref{prop:ordered_reduction_wo_copies}, but with two rounds.
First, we will only use the $m$ `original' goods to determine the picking sequence, this is identical to the picking sequence presented earlier. 
Then, another picking sequence will use the `duplications' (and so will involve agents $1,\dots,t$, where each picks one good). Each agent in the second round may pick any good in $M$ which was not yet allocated in this round and was not picked by him in the first round.
Let $A_i \subseteq M$ be the set of goods picked by agent $i$ in both rounds.

Any agent $i \in N \setminus \{1,\dots,t\}$ only participates in the first round, so as in \Cref{prop:ordered_reduction_wo_copies}, it holds that $v_i(A_i) \ge v'_i(A'_i)$.
For agent $j \in \{1,\dots,t\}$, let $B_j$ be the set of $|A_j|-1$ goods he picked in the first round. To show that $v_j(A_j) \ge v'_j(A'_j)$, it is enough to show that at the beginning of the second round there is a set of goods $H \subseteq M \setminus B_j$, such that (i) $|H| \ge j$, (ii) each $h \in H$ satisfies $v_j(B_j \cup \{h\}) \ge v'_j(A'_j)$. This would imply that in round $j$, player $j$ can pick a good $h$ such that $A_j = B_j \cup \{h\}$ satisfies the claim.

Denote $A'_j = \{j, j_1,\dots,j_k\}$, $B_j = \{b_1,\dots,b_k\}$, where $b_\ell$ was picked by agent $j$ in round $j_\ell$.
The existence of $H$ follows from applying \Cref{lem:important} with $X = \{\sigma_j(j_1),\dots,\sigma_j(j_k)\}$, $Y = B_j$, 
{$G = \{\sigma(g_1),\dots,\sigma_j(g_j)\}$} and $\tau = v'_j(A'_j)$.
Observe that by definition of the first picking sequence, for any $\ell$, it holds that $v_i(b_\ell) \ge v_i(\sigma_j(j_\ell)) = v'_j(j_\ell)$. 
Also, {$\tau = v'_j(A'_j) = v'_j(A'_j \setminus \{g_j\}) + v'_j(g_j) = v_j(X) + v_j(\sigma_j(g_j)) = v_j(X \cup \{\sigma_j(g_j)\})$}.
By definition, any good in $g \in G$ has {$v_j(\{g\}) \ge v_j(\sigma_j(g_j))$}, and by additivity all the conditions of \Cref{lem:important} are satisfied. This concludes the proof.
\end{proof}
For the rest of this section, we only consider ordered instances. Because the algorithms we use to solve these instances always output a simple allocation, their guarantees, by \Cref{prop:ordered_simple}, can be extended to any instance, even unordered.

\subsection{Reduction Rules}\label{sec:reductions}

\begin{definition} \label{def:alphavalid}
    Given an instance $\GenInstance$, an allocation rule that allocates bundles $A_1, \ldots, A_k$ to agents $1, \ldots, k$ respectively is $\alpha$-valid, if and only if the following holds:
    \begin{itemize}
        \item $v_i(A_i) \geq \alpha \cdot \mu_i(\calI)$ for all $i \in [k]$, and
        \item $\mu^{n-k}_i(M \setminus (A_1 \cup \ldots \cup A_k)) \geq \mu^n_i(M)$ for all $i \in N \setminus [k]$.
    \end{itemize}
\end{definition}

The following lemma with small value of $k$ (namely $k \leq 3$) has been proven and used in \cite{garg2021improved,akrami2023simplification,akrami2023breaking34barrierapproximate}.
\begin{restatable}{lemma}{reducK}\label{lem:reduceK}
    Given an ordered fair division instance $\GenInstance$, a fixed agent $i \in N$, and integer $k < |M|/n$, 
    let $K = \{k(n-1)+1, k(n-1)+2, \ldots, nk, nk+1\}$. Then, $\mu^{n-1}_i(M \setminus K) \geq \mu^n_i(M).$
\end{restatable}

\begin{tcolorbox}[%colback=blue!5,coltitle=blue!50!black,colframe=blue!25,
	title= \textbf{Rule $R^\alpha_k(\calI)$}]
	
	Preconditions:\label{pre1}
    \vspace{-2 mm}
	\begin{itemize}
		\item For $K = \{k(n-1)+1, \ldots, nk, nk+1\}$, there exists an agent $i$ such that $v_i(K) \geq \alpha \cdot \mu_i(\calI)$.
	\end{itemize}
    Process:
    \vspace{-2 mm}
    \begin{itemize}
        \item Allocate $K$ to agent $i$.
        \vspace{-2 mm}
		\item Set $N' \leftarrow N \setminus \{i\}$ and $M' \leftarrow M \setminus K$.
	\end{itemize}
    Guarantees:
    \vspace{-2 mm}
	\begin{itemize}
        \item For all $i \in N'$, $\mu_i(\calI') \geq \mu_i(\calI)$.
	\end{itemize}
\end{tcolorbox}

\begin{corollary}[of Lemma \ref{lem:reduceK}]\label{cor:validK}
    Given an ordered instance $\GenInstance$, for all $k < |M|/n$ and $\alpha \geq 0$, $R^\alpha_k(\calI)$ is an $\alpha$-valid reduction. 
\end{corollary}

\begin{definition}\label{def:irreducibility}
    We call an ordered instance $\GenInstance$, $R^\alpha_k$-irreducible, if $R^\alpha_k(\calI)$ is not applicable. Moreover, we call an ordered instance $\GenInstance$, $R^\alpha$-irreducible, if for all $k < |M|/|N|$, $R^\alpha_k(\calI)$ is not applicable. 
\end{definition}

\begin{restatable}{observation}{obsSimpleBounds}\label{obs:simple-bounds}
    If $\,\GenInstance$ is ordered and $R^\alpha_k$-irreducible for some $k < |M|/n$, then $v_i(nk+1) < \alpha \cdot \mu_i/(k+1)$ for all agents $i \in N$.
\end{restatable}

\begin{algorithm}[tb]
\caption{$\mathtt{R-reduce(\calI,\alpha)}$: 
\\ \textbf{Input:} An ordered instance $\GenInstance$ and approximation factor $\alpha$.
\\ \textbf{Output:} An ordered instance $\calI'=(N',M',\{v_i\}_{i \in N'})$ with $N' \subseteq N$ and $M' \subseteq M$.}
\label{alg:Reduce}\SetAlgoLined
\DontPrintSemicolon
\LinesNumbered

\While{for any $k < |M|/n$, $R^\alpha_k$ is applicable}{
     $\calI \leftarrow R^\alpha_k(\calI)$ for an arbitrary $k$ such that $R^\alpha_k$ is applicable}
\Return $(N,M,\{v_i\}_{i \in N})$
\end{algorithm}

\begin{restatable}{lemma}{reduceTwoGoods}\label{lem:reduce-2goods}
    Let $\GenInstance$ be a fair division instance with additive valuations, and fix an agent $i \in N$. Let $g_1, g_2 \in M$ be such that $v_i(g_1) + v_i(g_2) \leq \mu^n_i(M)$. Then 
    $\mu^n_i(M) \leq \mu^{n-1}_i(M \setminus \{g_1,g_2\}).$
\end{restatable}

\begin{restatable}{lemma}{newReduction}\label{lem:new-reduction}
    Given a fair division instance $\GenInstance$ and a fixed agent $i \in N$ with additive valuation, let $g_1, g_2, g_3 \in M$ be such that $g_2 \neq g_3$ and $v_i(g_2) + v_i(g_3) \leq \mu^n_i(M)$. Then 
    $\mu^{n-2}_i(M \setminus \{g_1,g_2,g_3\}) \geq \mu^n_i(M).$
\end{restatable}

\begin{tcolorbox}[%colback=blue!5,coltitle=blue!50!black,colframe=blue!25,
	title= \textbf{Rule $S^\alpha(\calI)$}]\label{rule:S}
	
	Preconditions:\label{pre1}
    \vspace{-2 mm}
	\begin{itemize}
        \item $\calI$ is $R^\alpha_1${-irreducible}.
        \vspace{-2 mm}
		\item There exists an agent $i$ such that $v_i(\{1,n+1\}) \geq \alpha \cdot \mu_i(\calI)$.
	\end{itemize}
    Process:
    \vspace{-2 mm}
    \begin{itemize}
        \item Allocate $\{1,n+1\}$ to $i$.
        \vspace{-2 mm}
        \item Set $N' \leftarrow N \setminus \{i\}$ and $M' \leftarrow M \setminus \{1,n+1\}$.
        \vspace{-2 mm}
        \item If there exists agent $j \neq i$ such that $v_j(\{1,n+1\}) \geq \alpha \cdot \mu_j(\calI)$: 
        \vspace{-1 mm}
            \subitem - Duplicate good $1$ and allocate $\{1,n\}$ to agent $j$.
            \vspace{-1 mm}
            \subitem - Set $N' \leftarrow N' \setminus \{j\}$ and $M' \leftarrow M' \setminus \{n\}$.
	\end{itemize}
    Guarantees:
    \vspace{-2 mm}
	\begin{itemize}
        \item For all $i \in N'$, $\mu_i(\calI') \geq \mu_i(\calI)$.
	\end{itemize}
\end{tcolorbox}

\begin{restatable}{lemma}{lemValidS}\label{lem:validS}
    Given an ordered instance $\GenInstance$ and $\alpha \leq 1$, $S^\alpha(\calI)$ is an $\alpha$-valid reduction. 
\end{restatable}

\begin{tcolorbox}[%colback=blue!5,coltitle=blue!50!black,colframe=blue!25,
	title= \textbf{Rule $T^\alpha(\calI)$}]\label{rule:T}
	
	Preconditions:\label{preT}
    \vspace{-2 mm}
	\begin{itemize}
        \item $\calI$ is $R^\alpha_k${-irreducible for $k \leq 1$.} 
        \vspace{-2 mm}
		\item There exists three different agents $i,j,\ell$ such that $v_i(\{1,n+1\}) \geq \alpha \cdot \mu_i(\calI)$, $v_j(\{2,n+2\}) \geq \alpha \cdot \mu_j(\calI)$, and $v_\ell(\{1,n+3\}) \geq \alpha \cdot \mu_\ell(\calI)$.
	\end{itemize}
    Process:
    \vspace{-2 mm}
    \begin{itemize}
        \item Allocate $\{1,n+1\}$ to $i$ and $\{2,n+2\}$ to $j$.
        \vspace{-2 mm}
        \item Duplicate good $1$ and allocate $\{1,n+3\}$ to $\ell$.
        \vspace{-2 mm}
        \item Set $N' \leftarrow N \setminus \{i,j,\ell\}$ and $M' \leftarrow M \setminus \{1,2,n+1,n+2,n+3\}$.
	\end{itemize}
    Guarantees:
    \vspace{-2 mm}
	\begin{itemize}
        \item {If $\alpha \leq 4/5$,} for all $i \in N'$, $\mu_i(\calI') \geq \mu_i(\calI)$.
	\end{itemize}
\end{tcolorbox}

\begin{restatable}{lemma}{fourOverFive}\label{lem:4-5-reduction}
    Let $\calI$ be a fair division instance with additive valuations and let $S = \{h_1, h_2, s_1, s_2, s_3 \} \subseteq M$ be such that for $j \in [2]$, $v_i(h_j) \leq 4/5 \cdot \mu^n_i(M)$ and for $j \in [3]$, $v_i(s_j) \leq 2/5 \cdot \mu^n_i(M)$ for all $i \in  N$. Then for all $i \in  N$ we have
    $\mu^{n-3}_i(M \setminus S) \geq \mu^n_i(M).$
\end{restatable}

\begin{proof}
    To prove the lemma, it suffices to partition (a subset of) $M \setminus S$ into $n-3$ bundles, each of value at least $\mu^n_i(M)$ to $i$. We call such a partition a certificate. Let $P=(P_1, \ldots, P_n)$ be an MMS partition of agent $i$ in the original instance $\calI$. Without loss of generality, assume $S \subseteq P_1 \cup \ldots \cup P_k$ for $k \leq 5$. If $S \subseteq P_1 \cup P_2 \cup P_3$, then $(P_4, \ldots, P_{n-1}, P_n)$ is a certificate. Otherwise, we consider two cases:
    \begin{itemize}
        \item Case 1: $S \subseteq P_1 \cup \ldots \cup P_4$. We have
        \begin{align*}
            v_i(P_1 \cup \ldots \cup P_4 \setminus S) &\geq \mu^n_i(M)(4 - 2 \cdot 4/5 - 3 \cdot 2/5) > \mu^n_i(M).  
        \end{align*}
        Therefore, $(P_1 \cup \ldots \cup P_4 \setminus S, P_5, \ldots, P_n)$ is a certificate.
        \item Case 2: $S \subseteq P_1 \cup \ldots \cup P_5$. Let $P_j \cap S = \{s_j\}$ for $j \in [3]$ and $P_{j+3} \cap S = \{h_j\}$ for $j \in [2]$. We have 
        \begin{align*}
            v_i(P_1 \cup P_2 \setminus S) \geq \mu^n_i(M)(2 - 2 \cdot 2/5) > \mu^n_i(M),
        \end{align*}
        and 
        \begin{align*}
            v_i(P_3 \cup P_4 \cup P_5 \setminus S) \geq \mu^n_i(M)(3 - 2 \cdot 4/5 - 2/5) = \mu^n_i(M).
        \end{align*}
        Thus, $(P_1 \cup P_2 \setminus S, P_3 \cup P_4 \cup P_5 \setminus S, P_6, \ldots, P_n)$ is a certificate. \qedhere
    \end{itemize}
\end{proof}

\begin{restatable}{lemma}{validT}\label{lem:validT}
    Given an ordered instance $\GenInstance$ and $\alpha \leq 4/5$, $T^\alpha(\calI)$ is an $\alpha$-valid reduction. 
\end{restatable}
\begin{proof}
    For all agents $i$ who is assigned a bundle $B$ in $T^\alpha$, we have $v_i(B) \geq \alpha \cdot \mu_i(\calI)$. {Note that by the preconditions of the reduction $T^\alpha(\calI)$, $\calI$ is $R^\alpha_k$-irreducible for $k \leq 1$. Thus, by Observation \ref{obs:simple-bounds},} for all the remaining agents $i$, we have $\alpha \cdot \mu^n_i(M) \geq v_i(1) \geq v_i(2)$ and $\alpha/2 \cdot \mu^n_i(M)\geq  v_i(n+1) \geq v_i(n+2) \geq v_i(n+3)$. Since $\alpha \leq 4/5$, by Lemma \ref{lem:4-5-reduction}, $v^{n-3}_i(M \setminus \{1,2,n+1,n+2,n+3\}) \geq v^n_i(M)$.    
\end{proof}

\begin{remark}
    In subsequent sections, similar to \Cref{def:irreducibility}, an ordered instance is $S^\alpha$-irreducible and $T^\alpha$-irreducible, if $S^\alpha$ and $T^\alpha$ are not applicable respectively.
\end{remark}

\subsection{Combining Round Robin and Bag-Filling}\label{sec:RR}
In this section, we present Algorithm \ref{alg:RR} ($\bagfillRR(\calI,\alpha)$) which combines the idea of bag-filling and round robin to achieve an $\alpha$-MMS allocation with copies for ordered instances.

Note that the input $\calI$ of $\bagfillRR(\calI,\alpha)$ is $R^\alpha$-irreducible. 
For all the agents $i$, we define $\mu_i = \mu^n_i(M)$ to be the MMS value of agent $i$ in the beginning of the Algorithm \ref{alg:RR}. We do a round robin fashion process as following: starting from $n$ empty bags $B_1, \ldots, B_n$, for $j$ in the cyclic sequence $1,2, \ldots, n$, we add the most valuable available  good (not yet allocated to any bag) to $B_j$. If for some remaining agent $i$, $v_i(B_j) \geq \alpha \cdot \mu_i$, then we allocate $B_j$ to $i$ and remove $i$ from the set of agents and $j$ from the cyclic sequence. We call this phase of the algorithm, the round robin phase. Note that for all the remaining bags $B_a$ and $B_b$, if $a<b$, then $v_i(B_a) \geq v_i(B_b)$ for all agents $i$ and hence, first index $1$ will be removed, then index $2$ and so on (see Lemma \ref{lem:order}). Let $\{B_a, \ldots, B_n\}$ be the set of the remaining bags at the end of the round robin phase. In the next phase of the algorithm, called duplication phase, we duplicate goods $1,2, \ldots, n-a+1$, add them to $B_a, \ldots, B_n$ respectively and allocate these bags to the remaining agents arbitrarily. Unlike the round robin phase, it is not clear whether the agents who are allocated during the duplication phase value their bundle at least $\alpha$ fraction of their MMS. However, in Sections \ref{sec:67} and \ref{sec:45}, we show that by appropriately reducing the instance in advance, we can ensure both the desired approximation guarantees and the bound on the number of distinct copies. 

\begin{algorithm}[tb]
\caption{$\bagfillRR(\calI,\alpha)$: 
\\ \textbf{Input:} An $R^\alpha$-irreducible instance $\GenInstance$ and approximation factor $\alpha$
\\ \textbf{Output:} An $\alpha$-MMS allocation $A$ with copies}
\label{alg:RR}\SetAlgoLined
\DontPrintSemicolon
\LinesNumbered
Let $B_i = \emptyset$ for $i \in [|N|]$ \;
$n \leftarrow |N|$, $j \leftarrow 1$, 
$a \leftarrow 1$ \;
\For(\Comment{Round robin Phase}){$g \gets 1$ \KwTo $|M|$} {
    $B_j \leftarrow B_j \cup \{g\}$ \;
    \If{there exists $i \in N$ such that $v_i(B_j) \geq \alpha \cdot \mu_i$}{
        $A_i \leftarrow B_j$ for such an arbitrary $i$ \;
        $N \leftarrow N \setminus \{i\}$, $ a \leftarrow a+1$ \;
    } 
    $j \leftarrow j+1$ \;
    \If{$j > n$}{
        $j \leftarrow a$ \;
    }
}
\For(\Comment{Duplication Phase}){$j \gets a$ \KwTo $n$}{  
    Let $i \in N$ \;
    $A_i \leftarrow B_j \cup \{j-a+1\}$, $N \leftarrow N \setminus \{i\}$ \;
}
\Return $(A_1, A_2, \ldots, A_n)$ \;
\end{algorithm}
For the rest of this (sub-)section, let $\alpha \geq 0$, $\calI$ be an $R^\alpha$-irreducible instance, {and $\mu_i(\calI)=1$ without loss of generality}.

\begin{lemma}\label{lem:upper-bound}
    For all the remaining agents $i$ after the round robin phase of $\bagfillRR(\calI,\alpha)$, and all the bags $B_j$ with $|B_j| \geq 3$ that were allocated to other agents during the round robin phase, we have $v_i(B_j) \leq 4\alpha/3$.
\end{lemma}
\begin{proof}
    Let $g$ be the last good added to $B_j$. 
    We have $g \geq 2n+1$ and hence $v_i(g) \leq v_i(2n+1) \leq \alpha/3$. The last inequality holds since the instance is $R^\alpha_2$-irreducible (by Observation \ref{obs:simple-bounds}). We have $v_i(B_j \setminus \{g\}) < \alpha$. Otherwise, $B_j \setminus \{g\}$ would have been allocated to $i$, or some other agent who values it at least $\alpha$, previously. Therefore we have, 
    \begin{align*}
        v_i(B_j) &= v_i(B_j \setminus \{g\}) + v_i(g) \leq \alpha + \alpha/3 = 4\alpha/3. \qedhere
    \end{align*}
\end{proof}

\begin{lemma}\label{lem:order}
    Let $\calB$ be the set of bags allocated during the round robin phase of $\bagfillRR(\calI,\alpha)$. Then $\calB = \{B_1, \ldots, B_{|\calB|}\}$.
\end{lemma}
\begin{proof}
    We prove that first $B_1$ gets allocated, then $B_2$ and so on. 
    Towards contradiction, assume for some $j>1$ bag $B_j$ gets allocated to agent $i$ while bag $B_{j-1}$ is not assigned. At the time that $B_j$ gets assigned, we have $|B_{j-1}| = |B_j|$ since both $B_{j-1}$ and $B_j$ were present in all the rounds so far. Let $B_{j-1}=\{g_1, \ldots, g_k\}$ and $B_j = \{g'_1, \ldots, g'_k\}$. Followed from round robin, for all $\ell \in [k]$, we have $g_\ell = g'_\ell-1$ which means $v_i(g_\ell) \geq v_i(g'_\ell)$. Hence $v_i(B_{j-1}) \geq v_i(B_j) \geq 1$ which is a contradiction that $B_{j-1}$ was not assigned to anyone before moving to $B_j$.
\end{proof}

\subsection{\boldmath $6/7$-MMS with Copies}\label{sec:67}
In this section, we present an algorithm which given an ordered instance, outputs a simple $\alpha$-MMS allocation with at most $\floor{n/2}$ distinct copies for any $\alpha \leq 6/7$. Then, using \Cref{prop:ordered_simple}, we prove the following theorem.

\thmnoverTwo*

First we reduce the instance in Algorithm \ref{alg:ReduceWithCopies} using $R^\alpha_0, R^\alpha_1, \ldots,$ and $S^\alpha$. It is not difficult to see that if $k$ agents are removed (i.e. are assigned a bundle) by the end of the algorithm, then in total at most $\floor{k/2}$ duplicated goods are introduced, each good is duplicated at most once, {and all the instances of these goods (original and copy) are allocated to different agents}. Also, since we remove these $k$ agents and all the goods in these $k$ bags, in the future no more copies of them can be used. Since all these rules are $\alpha$-valid reductions, it suffices to give {a simple} $\alpha$-MMS allocation with at most $\floor{(n^*-k)/2}$ distinct copies for the instance obtained by the remaining agents and goods where $n^*$ is the original number of agents. Let $n = n^*-k$.

Then, we run Algorithm \ref{alg:RR} ($\bagfillRR(\calI,\alpha)$) on the reduced instance. We prove that by the end of the round robin phase, at most $\lfloor n/2 \rfloor$ agents remain. In other words, if the bags $B_a, B_{a+1}, \ldots, B_n$ are the remaining bags, then $a \geq \floor{n/2}+1$. Furthermore, we prove for all remaining agents $i$, $v_i(B_j \cup \{n-a+1\}) \geq \mu_i(\calI)$ for all $j \in \{a, \ldots, n\}$. Therefore, to ensure MMS for the remaining agents, it suffices to duplicate goods $1,2, \ldots, n-a+1 \leq \floor{n/2}<a$, add them to $B_a, \ldots, B_n$ respectively and allocate these bags to the remaining agents arbitrarily. This is exactly the duplication phase of $\bagfillRR(\calI,\alpha)$. In this phase, we duplicate at most $\floor{n/2}$ many goods and at most one copy of each good. 

\begin{algorithm}[tb]
\caption{$\mathtt{S-reduce(\calI,\alpha)}$: 
\\ \textbf{Input:} An ordered instance $\GenInstance$ and approximation factor $\alpha$
\\ \textbf{Output:} An ordered instance $\calI'=(N',M',\{v_i\}_{i \in N'})$ with $N' \subseteq N$ and $M' \subseteq M$.}
\label{alg:ReduceWithCopies}\SetAlgoLined
\DontPrintSemicolon
\LinesNumbered

$\calI \leftarrow \mathtt{R-reduce(\calI)}$ \;
\While{$S^\alpha(\calI)$ is applicable}{
    $\calI \leftarrow S^\alpha(\calI)$ \;
    $\calI \leftarrow \mathtt{R-reduce(\calI)}$ \;
}
\Return $\calI = (N,M,\{v_i\}_{i \in N})$ \;
\end{algorithm}

Lemma \ref{lem:s-reduc-alg} follows from Lemma \ref{lem:validS} and Corollary \ref{cor:validK}.
\begin{lemma}\label{lem:s-reduc-alg}
    For an arbitrary ordered instance $\calI$ and $\alpha \leq 1$, let $\calI' = \mathtt{S-reduce(\calI,\alpha)}$. Then for all agents $i \in N'$, $\mu_i(\calI') \geq \mu_i(\calI)$. 
\end{lemma} 

For the rest of this (sub)-section, let $\alpha \leq 6/7$, $\calI$ be an $R^\alpha$-irreducible and $S^\alpha$-irreducible instance, and $\mu_i(\calI)=1$. Let $n$ be the number of agents right before the round robin phase of $\bagfillRR(\calI,\alpha)$, $\calB = \{B_1, \ldots, B_{a-1}\}$ be the set of bags allocated during this phase, and $n' = n-a+1$ be the number of remaining agents after this phase.

\begin{restatable}{observation}{obsSizeThree}\label{obs:size3}
    For all bags $B_j \in \calB$, $|B_j| \geq 3$.
\end{restatable}
\begin{proof}
    Since $\calI$ is $S^\alpha$-irreducible, $v_i(\{j,n+j\}) \leq v_i(\{1,n+1\})<\alpha$ for all $i \in N$. Thus, if $B_j$ was allocated to some agent, $|B_j|>2$. 
\end{proof}

\begin{lemma}\label{lem:many-valuable-bags}
    Let $i$ be a remaining agent after the round robin phase of $\bagfillRR(\calI,\alpha)$ (if any). Then $|\{B \in \calB \mid v_i(B) > 1\}| > n'$.
\end{lemma}
\begin{proof}
    Let $\mathcal{B}^+ = \{B \in \mathcal{B} \mid v_i(B) > 1\}$. By definition, for all $B \in \mathcal{B} \setminus \mathcal{B}^+$, $v_i(B) \leq 1$. By Observation \ref{obs:size3}, for all $B \in \mathcal{B}$ , $|B| \geq 3$. By Lemma \ref{lem:upper-bound}, for all $B \in \mathcal{B}^+$, $v_i(B) \leq 4\alpha/3$. Also, let $\bar{\mathcal{B}}$ be the set of $n-|\calB| = n'$ remaining bags after the round robin phase. For all $B \in \bar{\mathcal{B}}$, $v_i(B) < \alpha$. Otherwise, $B$ should have been allocated to $i$ (or some other agent who values it at least $\alpha$). Therefore, we have
    \begin{align*}
        n &\leq v_i(M) \\
        &= \sum_{B \in \mathcal{B}^+} v_i(B) + \sum_{B \in \mathcal{B} \setminus \mathcal{B}^+} v_i(B) + \sum_{B \in \bar{\mathcal{B}}} v_i(B) \\
        &< |\mathcal{B}^+| 4\alpha/3 + (|\mathcal{B}| - |\mathcal{B}^+|) + (n - |\mathcal{B}|) \alpha \\
        &\leq (|\mathcal{B}| - |\mathcal{B}^+|) + 6/7 (|\mathcal{B}^+|/3 + n - |\calB| + |\mathcal{B^+}|) \tag{$\alpha \leq 6/7$}
    \end{align*}
    Removing $(|\mathcal{B}| - |\mathcal{B}^+|)$ from both sides of the inequality, we obtain, 
    \begin{align*}
        n - |\calB| + |\calB^+| &< 6/7 (|\calB^+|/3 + n - |\calB| + |\calB^+|) = 2|\calB^+|/7 + 6/7(n - |\calB| + |\calB^+|).
    \end{align*}
    Therefore, $|\calB^+| > (n - |\calB| + |\calB^+|)/2$ which means $|\calB^+| > n - |\calB|=n'$.
\end{proof}

\begin{corollary}[of Lemma \ref{lem:many-valuable-bags}]\label{cor:halfRemaining}
     $n' \leq \floor{n/2}$.        
\end{corollary}
\begin{lemma}\label{lem:duplic-suffice}
    For all the remaining agents $i$ after the round robin phase of $\bagfillRR(\calI,\alpha)$, and all the remaining bags $B_j$, we have $v_i(B_j \cup \{n'\}) \geq 1$.
\end{lemma}
\begin{proof}
    By Lemma \ref{lem:order}, the bags $B_1, \ldots, B_{n-n'}$ are allocated during the round robin phase. 
    By Lemma \ref{lem:many-valuable-bags}, there exists at least $n'$ many bags $B_{i_1}, \ldots, B_{i_{n'}}$ for $i_1 < i_2 < \ldots < i_{n'} \leq n-n'$ that were allocated to other agents during the round robin phase and $v_i(B_{i_\ell}) > 1$ for all $\ell \in [n']$. Note that $i_{n'} \geq n'$. Since $B_j$ have been assigned a good in all but maybe the last round, we have $|B_j| \geq |B_{i_{n'}}|-1$. Let $B_j = \{g_1, \ldots, g_k\}$ and $B_{i_{n'}} = \{g'_1, \ldots, g'_{k'}\}$ for $j= g_1 < \ldots < g_k$ and $i_{n'} = g'_1 < \ldots < g'_{k'}$. For all $\ell \in [k'-1]$, since $g_\ell$ is allocated before $g'_{\ell+1}$, $g_\ell < g'_{\ell+1}$ and hence $v_i(g_\ell) \geq v_i(g'_{\ell+1})$. We have
    \begin{align*}
        v_i(B_j \cup \{n'\}) &= v_i(B_j) + v_i(n') \geq v_i(B_j) + v_i(i_{n'}) = \sum_{\ell \in [k]} v_i(g_\ell) + v_i(i_{n'}) \\
        &\geq \sum_{\ell \in [k'-1]} v_i(g'_{\ell+1}) + v_i(i_{n'}) = v_i(B_{i_{n'}}) \geq 1. \qedhere
    \end{align*}
\end{proof}

\begin{proof}[Proof of \Cref{thm:nover2copiesMMS}]
    We prove that for ordered instances $\GenInstance$, $$\bagfillRR(\mathtt{S-reduce(\calI,\alpha)},\alpha)$$ returns {a simple} $\alpha$-MMS allocation with at most $\floor{n^*/2}$ many distinct copies. Then, from \Cref{prop:ordered_simple}, the theorem follows.
    
    Let $n$ be the number of agents after applying $\mathtt{S-reduce(\calI,\alpha)}$. Since $S^\alpha$ and $R^\alpha_k$ are valid reductions for all $k < |M|/n^*$, all agents who receive a bag during the $\mathtt{S-reduce(\calI,\alpha)}$ receive at least $\alpha$ fraction of their MMS value and the MMS value of the remaining agents do not decrease after this process. {Without loss of generality, let the MMS value of the remaining agents be $1$ at this point of the algorithm.} All the agents who receive a bag during the round robin phase of $\bagfillRR$ get at least $\alpha$. Let agent $i$ be an agent who receives bag $B_j$ after the round robin phase. Let $g$ be the good that was duplicated and added to $B_j$ right before allocating it to agent $i$. $g \leq n'$ and thus $v_i(g) \geq v_i(n')$. By Lemma \ref{lem:duplic-suffice}, $v_i(B_j \cup \{g\}) \geq \alpha$. Therefore, the final allocation is $\alpha$-MMS. 
    
    We now prove that the generated MMS allocation is simple, namely that there are $t \le \floor{n^*/2}$ distinct copies,
    and that the $2t$ instances, composed of these copies and their corresponding original goods, are allocated to $2t$ distinct agents.
    Let $k=n^*-n$ be the number of agents who receive a bag in $\mathtt{S-reduce(\calI,\alpha)}$. Only when applying $S^\alpha$, new duplications are introduced and one duplication per two agents that are removed. So during $\mathtt{S-reduce(\calI,\alpha)}$, at most $\floor{k/2}$ distinct copies are used. Each instance of these good is allocated to a different agent. Furthermore, all the instances of the goods that are allocated during $S^\alpha$ along with their owners, are removed from the instance. Therefore, these goods cannot be copied more than once and their owners cannot get any other goods later on. In the duplication phase of $\bagfillRR$, we duplicate items $1,2, \ldots, n'$ once and allocate the original and the duplications to $2n'$ different agents. By Corollary \ref{cor:halfRemaining}, $n' \leq \floor{n/2}$. Therefore, in total at most $\floor{n^*/2}$ distinct copies are used and the allocation is simple. 
\end{proof}

\subsection{\boldmath $4/5$-MMS with Copies}\label{sec:45}
In this section, we present an algorithm which, given an ordered instance, outputs {a simple} $\alpha$-MMS allocation with at most $\floor{n/3}$ distinct copies for any $\alpha \leq 4/5$. Then, using \Cref{prop:ordered_simple}, we prove the following theorem in Appendix \ref{sec:approx-additive_proofs}.

\thmnOverThree*

Similar to Section \ref{sec:67}, first we reduce the instance but this time by running \Cref{alg:ReduceWith3Copies} which uses  $R^\alpha_0, R^\alpha_1, \ldots,$ and $T^\alpha$. 
Analogous to \Cref{alg:ReduceWithCopies}, it is not difficult to see that if $k$ agents are removed (i.e. are assigned a bundle) by the end of \Cref{alg:ReduceWith3Copies}, then (i) at most $\floor{k/3}$ duplicated goods are introduced, (ii) each good is duplicated at most once, and (iii) all the duplicated goods and their corresponding original instances are allocated to distinct agents. 
Also, since we remove these $k$ agents and all the goods in these $k$ bags, in the future no more copies of these goods can be used and no other goods can be allocated to these agents. Since all these rules are $\alpha$-valid reductions, it suffices to give a simple $\alpha$-MMS allocation with at most $\floor{(n^*-k)/3}$ distinct copies for the instance obtained by the remaining agents and goods where $n^*$ is the original number of agents. Let $n = n^*-k$.

Then, we run Algorithm \ref{alg:RR} ($\bagfillRR(\calI,\alpha)$) on the reduced instance. Following a similar idea, but a more involved calculation, we prove that we end up with a $4/5$-MMS allocation using at most $\floor{n/3}$ distinct copies.

\begin{algorithm}[tb]
\caption{$\mathtt{T-reduce(\calI,\alpha)}$: 
\\ \textbf{Input:} An ordered instance $\GenInstance$ and approximation factor $\alpha$
\\ \textbf{Output:} An ordered instance $\calI'=(N',M',\{v_i\}_{i \in N'})$ with $N' \subseteq N$ and $M' \subseteq M$.}
\label{alg:ReduceWith3Copies}\SetAlgoLined
\DontPrintSemicolon
\LinesNumbered

$\calI \leftarrow \mathtt{R-reduce(\calI)}$ \;
\While{$T^\alpha(\calI)$ is applicable}{
    $\calI \leftarrow T^\alpha(\calI)$ \;
    $\calI \leftarrow \mathtt{R-reduce(\calI)}$ \;
}
\Return $\calI = (N,M,\{v_i\}_{i \in N})$ \;
\end{algorithm}

\begin{lemma}
    For an arbitrary ordered instance $\calI$ and $\alpha \leq 4/5$, let $\calI' = \mathtt{T-reduce(\calI,\alpha)}$. Then for all agents $i \in N'$, $\mu_i(\calI') \geq \mu_i(\calI)$. 
\end{lemma} 
\begin{proof}
    The lemma follows from Lemma \ref{lem:validT} and Corollary \ref{cor:validK}.
\end{proof}

For the rest of this (sub)-section, let $\alpha \leq 4/5$, $\calI$ be an $R^\alpha$-irreducible and $T^\alpha$-irreducible instance, and $\mu_i(\calI)=1$. Let $n$ be the number of agents right before the round robin phase of $\bagfillRR(\calI,\alpha)$, $\calB = \{B_1, \ldots, B_{a-1}\}$ be the set of bags allocated during this phase, and $n' = n-a+1$ be the number of remaining agents after this phase.

\begin{lemma}\label{obs:size3-2}
    For all bags $B_j \in \calB$ with $j > 2$, $|B_j| \geq 3$.
\end{lemma}
\begin{proof}
    Towards contradiction, assume for some $j>2$, $|B_j| \leq 2$ and $B_j$ is allocated to agent $i$. Note that since $\calI$ is $R_0^\alpha$-irreducible, $v_a(1) < \alpha$ for all agents $a$ and hence $|B| \geq 2$ for all $B \in \calB$. So for all $\ell \in [n]$, we have $\{\ell,n+\ell\} \in B_\ell$ and $B_j = \{j,n+j\}$. Note that for all $j'<j$, $v_i(\{j',n+j'\}) > v_i(\{j,n+j\}) \geq \alpha$. Hence, for all $j'<j$, $|B_{j'}|=2$. Let $a_1,a_2,a_3$ be the three agents who received $B_1, B_2, B_3$ respectively. Note that the precondition of Rule $T^\alpha$ is met for these agents which is a contradiction with $\calI$ being $T^\alpha$-irreducible.    
\end{proof}
\begin{observation}\label{obs:1-2}
    For $n\geq2$, $\{B_1,B_2\} \subseteq \calB$ and $v_i(B_j) \leq 3\alpha/2$ for all $i \in N$ and $j \in [2]$.
\end{observation}
\begin{proof}
    Since the instance is $R^\alpha_0$-irreducible, $v_i(2) \leq v_i(1) < \alpha$ {for all $i \in N$}. Therefore, $|B_j| \geq 2$. Let $g$ be the last good added to $B_j$. Since $g \geq n+1$, $v_i(g) \leq v_i(n+1) < \alpha/2$. We have 
    \begin{align*}
        v_i(B_j) &= v_i(B_j \setminus \{g\}) + v_i(g) < \alpha + \alpha/2 = 3\alpha/2. 
    \end{align*}
    Now assume $B_2 \notin \calB$. {Let $i$ be an agent who is not allocated any bundle in $\calB$.} Note that for all $B \notin \calB$, $v_i(B) < \alpha$. Otherwise, $B$ should have been allocated during the round robin phase. Therefore
    \begin{align*}
        n &\leq v_i(M) \\
        &= \sum_{j}v_i(B_j) \\
        &< 3\alpha/2 + (n-1)\alpha \\
        &\leq (4/5n+2/5) \cdot, \tag{$\alpha \leq 4/5$} 
    \end{align*}
    which is a contradiction with $n \geq 2$.
\end{proof}
\begin{lemma}\label{lem:many-valuable-bags-2}
    Let $i$ be a remaining agent after the round robin phase of $\bagfillRR(\calI,\alpha)$ (if any). Then $|\{B \in \calB \mid v_i(B) > \alpha\}| \geq 2n/3$ or $n\leq 5$.
\end{lemma}
\begin{proof}
    Assume $n \geq 6$. Let $\calB^+ = \{B \in \calB \mid v_i(B) > \alpha\}$. By definition, for all $B \in \calB \setminus \calB^+$, $v_i(B) \leq \alpha$. Let $\bar{\calB}$ be the set of $n-|\calB|=n'$ many remaining bags after the round robin phase. For all $B \in \bar{\calB}$, $v_i(B) < \alpha$. Otherwise, $B$ should have been allocated to $i$ (or some other agent who values it at least $\alpha$ fraction of their MMS value). We consider two cases:
    \paragraph{\boldmath Case 1: $v_i(B_2) \leq 4\alpha/3$.}
    For all $B \in \calB \setminus \{B_1,B_2\}$, Lemma \ref{obs:size3-2} implies $|B| \geq 3$ and Lemma \ref{lem:upper-bound} implies $v_i(B) \leq 4\alpha/3$. Also by Observation \ref{obs:1-2}, $v_i(B_1) \leq 3\alpha/2$ and we have $v_i(B_2) \leq 4\alpha/3$. Therefore
    \begin{align*}
        \sum_{B \in \calB^+} v_i(B) &\leq 4\alpha/3(|\calB^+|-1) + 3\alpha/2 \\
        &\leq ( 16|\calB^+|/15 + 2/15). \tag{$\alpha \leq 4/5$}
    \end{align*}
    
    We have
    \begin{align*}
        n &\leq v_i(M) \\
        &= \sum_{B \in \calB^+} v_i(B) + \sum_{B \in \calB \setminus \calB^+} v_i(B) + \sum_{B \in \bar{\calB}} v_i(B) \\
        &\leq (16|\calB^+|/15 + 2/15) + (|\calB| - |\calB^+|) \alpha + (n - |\calB|) \alpha \\
        &\leq (2/15 + 4n/5 + 4|\calB^+|/15). \tag{$\alpha \leq 4/5$}
    \end{align*}
    We obtain
    \begin{align*}
        n/5 &\leq 4|\calB^+|/15 + 2/15.
    \end{align*}
    Therefore, $|\calB^+| \geq (3n-2)/4 \geq 2n/3$ {(for $n \geq 6$)}.
    \paragraph{\boldmath Case 2: $v_i(B_2) > 4\alpha/3$.} If $|B_2| \geq 3$, then by Lemma \ref{lem:upper-bound}, $v_i(B_2) \leq 4\alpha/3$. Also if $|B_2|=1$, $v_i(B_2) < \alpha$ since the instance is $R^\alpha_0$ irreducible. Therefore, $|B_2|=2$. It implies $|B_1|=2$ and we have $B_1 = \{1,n+1\}$ and $B_2 = \{2,n+2\}$. Let $a$ and $b$ be the agents who received $B_1$ and $B_2$. If $v_i(\{1,n+3\}) \geq \alpha$, then the instance is not $T^\alpha$-irreducible since $a,b,i$ satisfy its precondition. Therefore, $v_i(1)+v_i(n+3) < \alpha$. 
    \begin{align*}
        v_i(1) + v_i(n+2) &\geq v_i(2) + v_i(n+2) = v_i(B_2) > 4\alpha/3.
    \end{align*}
    Therefore, 
    \begin{align*}
        v_i(n+3) &< v_i(n+2) - \alpha/3 \\
        &< (\alpha/2 - \alpha/3) \tag{$v_i(n+2) \leq v_i(n+1) < \alpha/2$ due to Observation \ref{obs:simple-bounds}}\\
        &= \alpha/6.
    \end{align*}
    Now fix an index $j>2$ and let $g$ be the last good added to $B_j$. We have 
    \begin{align*}
        v_i(B_j) &= v_i(B_j \setminus \{g\}) + v_i(g) \\
        &< \alpha + \alpha/6 \\
        &= 7\alpha/6.
    \end{align*}
    Therefore,
    \begin{align*}
        n &\leq v_i(M) \\
        &= v_i(B_1) + v_i(B_2) + \sum_{B \in \calB \setminus \{B_1,B_2\}} v_i(B) + \sum_{B \in \bar{\calB}} v_i(B) \\
        &< 3\alpha + (|\calB| - 2) 7\alpha/6 + (n - |\calB|) \alpha \\
        &\leq (4n/5 + (2|\calB|+8)/15) \tag{$\alpha \leq 4/5$} \\
        &\leq (14n/15 + 4/15). \tag{$|\calB| \leq n-2$}
    \end{align*}
    % }
    We obtain that $n < 4$ which is a contradiction with $n \geq 6$.
\end{proof}
    
\begin{corollary}[of Lemma \ref{lem:many-valuable-bags-2}]\label{cor:halfRemaining}
     $n' \leq \floor{n/3}$ for $n \geq 6$.        
\end{corollary}

\begin{lemma}\label{lem:duplic-suffice-2}
    If $n \geq 6$, for all the remaining agents $i$ after the round robin phase of $\bagfillRR(\calI,\alpha)$, and all the remaining bags $B_j$, we have $v_i(B_j \cup \{n'\}) \geq \alpha$.
\end{lemma}
\begin{proof}
    By Lemma \ref{lem:order}, the bags $B_1, \ldots, B_{n-n'}$ are allocated during the round robin phase. 
    By Lemma \ref{lem:many-valuable-bags-2}, there exists at least $2n/3 \geq n'$ many bags $B_{i_1}, \ldots, B_{i_{n'}}$ for $i_1 < i_2 < \ldots < i_{n'} \leq n-n'$ that were allocated to other agents during the round robin phase and $v_i(B_{i_\ell}) > \alpha$ for all $\ell \in [n']$. Note that $i_{n'} \geq n'$. Since $B_j$ have been assigned a good in all but maybe the last round, we have $|B_j| \geq |B_{i_{n'}}|-1$. Let $B_j = \{g_1, \ldots, g_k\}$ and $B_{i_{n'}} = \{g'_1, \ldots, g'_{k'}\}$ for $j= g_1 < \ldots < g_k$ and $i_{n'} = g'_1 < \ldots < g'_{k'}$. For all $\ell \in [k'-1]$, since $g_\ell$ is allocated before $g'_{\ell+1}$, $g_\ell < g'_{\ell+1}$ and hence $v_i(g_\ell) \geq v_i(g'_{\ell+1})$. We have
    \begin{align*}
        v_i(B_j \cup \{n'\}) &= v_i(B_j) + v_i(n') \\
        &\geq v_i(B_j) + v_i(i_{n'}) \\
        &= \sum_{\ell \in [k]} v_i(g_\ell) + v_i(i_{n'}) \\
        &\geq \sum_{\ell \in [k'-1]} v_i(g'_{\ell+1}) + v_i(i_{n'}) \\
        &= v_i(B_{i_{n'}}) > \alpha.
    \end{align*}
\end{proof}
 
\section{Discussion}\label{sec:discussion}

Herein, we show that exact MMS fairness can be guaranteed in settings that permit post facto adjustments to the supply of indivisible items. 
In particular, we obtain tight bounds for XOS valuations and beyond, and, in the context of $\alpha$-MMS, improve upon the state of the art using few copies.

The framework of resource augmentation provides a meaningful route for bypassing barriers in discrete fair division. In particular, it would be interesting to develop truthful mechanisms for MMS that leverage bounded supply adjustment. Another relevant direction for future work is to consider the framework for achieving envy-freeness up to any good (EFX). Such an approach would complement prior works on EFX with charity (see, e.g., \cite{chaudhury2021little}), which obtain existential guarantees for EFX with disposal, rather than duplication, of goods.

One may consider several metrics for measuring the augmentation in an allocation with copies. In this work we consider the number of copies made of each good and the total numbers of copies made ($k$ and $t$, respectively, in $(t,k)$-allocation).
Another possible measure is the total number of \emph{allocated} goods (including copies), required to achieve MMS. This notion is relevant for settings where unallocated goods are not manufactured (e.g., course allocation), or are refunded (e.g., medical supplies).
Budish~\cite{budish2011combinatorial} considered the $\ell_2$ norm of the $m$-dimensional vector whose components correspond to the number of copies made of the $m$ goods, respectively. Clearly one can consider many other measures which suggest a host of new problems.

\bibliographystyle{alpha}
\bibliography{refs.bib}

\appendix
\section{Proof of \Cref{thm:XOS-ordinal-impossible}}\label{appendix:ordinal-XOS}
Towards contradiction, assume there exists $f: \mathbb{N} \rightarrow \mathbb{N}$ such that $1$-out-of-$f(n)$ MMS allocations exist for all instances with $n$ agents. Let $n'=f(n)$ and consider the $n'$-dimensional cube and its corresponding XOS valuations as in \Cref{thm:lowermontone}. We call this instance $\mathcal{I} = \langle [n'], M, \{v_i\}_{i \in [n']} \rangle$. Let $N \subset [n']$ be the subset of the first $n$ agents in this example; $N = [n]$. Let $\GenInstance$. For all $i \in N$, $\mu^{f(n)}_i(M) = \mu_i(\mathcal{I}')$. Therefore, there exists an allocation $A=(A_1, \ldots, A_n)$ such that for all $i \in [n]$ 
$$v_i(A_i) \geq \mu^{f(n)}_i(M) = \mu_i(\mathcal{I}').$$
Hence, $P^1_i \subseteq A_1$ and $P^2_j \subseteq A_2$ for some $i,j \in [n']$, and $A_1 \cap A_2 = \emptyset$. However, $P^1_i \cap P^2_j = \{g_{d_1,\dots,d_n} \mid \; d_1 = i, d_2 = j \} \neq \emptyset$ which is a contradiction.

\section{Missing Proofs of \Cref{sec:exactMMS}}\label{apx:goods_proofs}
\lemBoundIndiv*
\begin{proof}
    For any agent $i$, let $P^i= (P^i_1,\dots,P^i_n)$ be his MMS partition.
    Consider the MMS allocation $A=(A_1,\dots,A_n)$ where bundle $A_i \sim \mathrm{Unif}(\{P^i_1,\dots, P^i_n\})$, independently for each agent $i$. That is, $i$'s bundle, $A_i$, is one of his MMS subsets chosen uniformly at random, independent of other choices.
    We use the balls-and-bins analogy \cite{raab1998balls}, and consider the process of constructing $A$ as throwing $n$ balls (the agents) into $n$ bins (the index in an MMS partition). We are interested in the  number of balls per bin, so we ignore the fact that the balls are distinct.

    Fix some good $g \in M$ and assume without loss that $g \in P^i_1$ for any $i \in N$. 
    The event in which $g$ is copied at least $k-1$ times for the allocation $A=(A_1,\dots,A_n)$, is exactly the event in which the bin corresponding to index $1$ has at least $k$ balls.
    A simple analysis shows that, 
    $$
    \Pr[g \text{ was copied at least }k-1 \text{ times}] \le \binom{n}{k}\left(\frac{1}{n}\right)^k \le  \left(\frac{e}{k}\right)^k,
    $$
    where the last transition follows from Stirling's approximation.

    If we plug in $k = \frac{3 \ln m}{\ln \ln m}$, we get
    \begin{align*}
        \left(\frac{e}{k}\right)^k %\\
        &=
        \left(\frac{e \ln \ln m}{3 \ln m}\right)^{\frac{3 \ln m}{\ln \ln m}} \\
        &\le
        \exp\left(\frac{3 \ln m}{\ln \ln m}\ln\left(\frac{e \ln \ln m}{3 \ln m}\right) \right)\\
        &\le
        \exp\left(\frac{3 \ln m}{\ln \ln m}\left(\ln \ln \ln m - \ln \ln m\right) \right)\\
        &=
        \exp\left(-3 \ln m + \frac{3 \ln m \ln \ln \ln m}{\ln \ln m} \right).
    \end{align*}

For large enough $m$ we have $\left(\frac{e}{k}\right)^k < \exp(-2 \ln m) = \frac{1}{m^2}$.

Applying the union bound we have,
\begin{align*}
    \Pr\left[\text{there exists a good with more than } \frac{3 \ln m}{\ln \ln m} \text{ copies}\right] < m \cdot \frac{1}{m^2} = \frac{1}{m},
\end{align*}
which completes the proof.
\end{proof}

\lemBagFillCopy*
\begin{proof}
Every time an agent is allocated, except the last agent, we duplicate the last item that was put in his bag and put the duplicate as the first item in the next bag. Since the value of every item is less than $1$ for all agents, every allocated bundle has at least two items. Thus, we duplicate a different item each time, and in total, $|N|-1$ items. We now show that upon termination, every agent gets at least a value of 1.

We show that for every $t$, for every $i\in N_t$, $v_i(M_t)\ge |N|-t$.  We prove by induction on $t$. Obviously, it's true for $t=0$. Assume this is true for $t-1\in [|N|-1]$, and consider the iteration when agent $i_{t-1}$ gets assigned bundle $A_{i_{t-1}}$. Consider some agent $i\in N_t$. By the inductive assumption, we have that for every $i'\in N_{t-1}$, $v_{i'}(M_{t-1})\ge |N|-t+1.$ As $N_t\subset N_{t-1}$, we have that 
\begin{eqnarray}
    v_i(M_{t-1})\ge |N|-t+1. \label{eq:assumption}    
\end{eqnarray}   
    
Now consider the bundle $B_{t-1}$ just before some item $j$ was added to it, causing agent $i_{t-1}$'s value to rise above 1. Since before adding $j$, the value each agent in $N_{t-1}$ assigned to $B_{t-1}$ was smaller than 1, we also have 
\begin{eqnarray}
    v_i(B_{t-1})< 1. \label{eq:step}    
\end{eqnarray}

This implies that
\begin{eqnarray*}
    v_i(M_{t})\ =\ v_i(M_{t-1}\setminus B_{t-1})\ =\ v_i(M_{t-1})- v_i(B_{t-1}) > |N|-t,
\end{eqnarray*}
where the last inequality follows Equations~\eqref{eq:assumption},~\eqref{eq:step}.
\end{proof}

\subsection{Example for Algorithm $\bagfillcopy$}\label{sec:examples}
\begin{figure}[tbh]
\begin{tabular}{|c|c|c|c|c|}
\hline
 & Agent 1 & Agent 2 & Agent 3 & Agent 4 \\ \hline
$g_1$ & 0.2     & 0.3     & 0.8     & 0.1     \\ 
$g_2$ & 0.2     & 0.5     & 0.3      & 0.1     \\ 
$g_3$ & 0.7     & 0.3     & 0.2     & 0.8     \\ 
$g_4$ & 0.3     & 0.2     & 0.3     & 0.2     \\ 
$g_5$ & 0.3     & 0.8     & 0.1     & 0.2     \\ 
$g_6$ & 0.5     & 0.5     & 0.1     & 0.7     \\ 
$g_7$ & 0.1     & 0.1     & 0.5     & 0.1     \\ 
$g_8$ & 0.1     & 0.1     & 0.7     & 0.2     \\ 
$g_9$ & 0.8     & 0.3       & 0.3     & 0.8     \\ 
$g_{10}$ & 0.2     & 0.8     & 0.1     & 0.3     \\ 
$g_{11}$ & 0.2     & 0.2     & 0.9     & 0.3     \\ 
$g_{12}$ & 0.7     & 0.3     & 0.1     & 0.5     \\ \hline
\end{tabular}
\caption{Running example for $\bagfillcopy$. There is an MMS allocation without copies for this instance. All MMS values are one. For agent 1 we have $v_1(\{g_1, g_2, g_3\})=1.1$, $v_1(\{g_4,g_5,g_6\})=1.1$, $v_1(\{g_7, g_8, g_9\})=1$ and $v_1(\{g_{10}, g_{11}, g_{12} \})=1.1$. For agent 2 we have $v_2(\{g_1, g_2, g_3\})=1.1$, $v_2(\{g_4,g_5\})=1$, $v_2(\{g_6.g_7, g_8, g_9\})=1$ and $v_2(\{g_{10}, g_{11}, g_{12} \})=1.1$. By construction the MMS partitions consist of contiguous runs of $g_1,\ldots,g_{12}$.}
\end{figure}

Consider the case when we add items in order $g_4, g_7, g_9, g_1, g_2, g_{10}, g_{12}, g_6, g_5, g_3,g_8, g_{11}$.
\begin{itemize}
\item For $B=\{g_4, g_7\}$, $v_1(B)=0.4$, $v_2(B)=0.3$, $v_3(B)=0.8$, and $v_4(B)=0.3$, no agent has attained the MMS value for this set $B$.
\item For $B=\{g_4, g_7, g_9\}$, $v_1(B)=1.2$, $v_2(B)=0.6$, $v_3(B)=1.1$, and $v_4(B)=1.1$, arbitrarily choose agent 1 and assign $B$ to agent 1, set $B$ to be a copy of the last good added: $B=\{g_9\}$.
\item For $B=\{g_9, g_1, g_2\}$, we have $v_2(B)=1.1$, $v_3(B)=1.4$, $v_4(B)=1$, arbitrarily choose agent 2 and assign $B$ to agent 2, set $B=\{g_2\}$.
\item For $B=\{g_2, g_{10}, g_{12}, g_6\}$, we have $v_3(B)=0.6$, $v_4(B)=1.6$ and assign $B$ to agent 4, set $B=g_6$.
\item For $B=\{g_6, g_5, g_3, g_8,g_{11}\}$ we have that $v_3(B)=2$ and assign $B$ to agent 3. 
\end{itemize}
Note that not all goods have been allocated, we could add the remaining goods to any of the agents, arbitrarily. 

The execution of $\bagfillcopy$ depends on the order in which the items are added.

\section{Missing Proofs of \Cref{approxMMS}}\label{sec:approx-additive_proofs}
\obsSimpleBounds*
\begin{proof}
    By irreducibility of $R^\alpha_k$ for $K = \{k(n-1)+1, k(n-1)+2, \ldots, nk, nk+1\}$, we have $v_i(K) < \alpha \cdot \mu_i$ for all $i \in N$. Since $v_i(nk+1) \leq v_i(nk) \leq \ldots \leq v_i(k(n-1)+1)$, we get $v_i(nk+1) \leq v_i(K)/(k+1) \leq \alpha \cdot \mu_i/(k+1)$.
\end{proof}

\reducK*
\begin{proof}
    Let $P=(P_1, P_2, \ldots, P_n)$ be an MMS partition of agent $i$ for the original instance. By the pigeonhole principle, there exists $j \in [n]$ such that $|P_j \cap [kn+1]| > k$. Without loss of generality, let us assume $|P_n \cap [kn+1]| > k$ and $g_1, g_2, \ldots g_{k+1}$ are $k+1$ distinct goods in $P_n \cap [kn+1]$ such that $g_1 \leq g_2 \leq \ldots \leq g_{k+1}$. For $j \in [k+1]$, let $kn-k+j \in P_{a_j}$. Now for all $j \in [k+1]$, swap the goods $kn-k+j$ and $g_j$; i.e., iteratively do the following:
    \begin{itemize}
        \item for all $j \in [k+1]$: $P_{a_j} \leftarrow P_{a_j} \setminus \{kn-k+j\} \cup \{g_j\}$ and $P_n \leftarrow P_n \setminus \{g_j\} \cup \{kn-k+j\}$.
    \end{itemize}
    Let $P'=(P'_1, P'_2, \ldots, P'_n)$ be the final partition. Since $g_1, \ldots, g_{k+1}$ are $k+1$ goods in $[kn+1]$ in increasing order of index, $g_j \leq kn-k+j$. Thus, $v_i(g_j) \geq v_i(kn-k+j)$ and  after each of these swaps, the value of $P_{a_j}$ cannot decrease. Therefore, for all $j \in [n-1]$, $v_i(P'_j) \geq v_i(P_j) \geq \mu^n_i(M)$ and $(P'_1, \ldots, P'_{n-1})$ is a partition of a subset of $M \setminus K$ into $n-1$ bundles of value at least $\mu^n_i(M)$ for $i$. The lemma follows.
\end{proof}

\reduceTwoGoods*
\begin{proof}
    To prove the lemma, it suffices to find a partition of (a subset of) $M \setminus \{g_1,g_2\}$ into $n-1$ bundles, such that the value of agent $i$ for each bundle is at least $\mu^n_i(M)$. We call such a partition a certificate. Let $P=(P_1, \ldots, P_n)$ be an MMS partition of agent $i$ in the original instance $\calI$. Without loss of generality, assume $\{g_1,g_2\} \subseteq P_1 \cup P_2$. If $\{g_1,g_2\} \subseteq P_1$, then $(P_2, \ldots, P_{n-1}, P_n)$ is a certificate. Otherwise, let $g_i \in P_i$ for $i \in [2]$. We have 
    \begin{align*}
        v_i(P_1 \cup P_2 \setminus \{g_1,g_2\}) 
        &\geq 2\mu^n_i(M) - \mu^n_i(M) \geq \mu^n_i(M).
    \end{align*}
    Hence $(P_1 \cup P_2 \setminus \{g_1,g_2\}, P_3, \ldots, P_n)$ is a certificate.
\end{proof}

\newReduction*
\begin{proof}
    To prove the lemma, it suffices to partition (a subset of) $M \setminus \{g_1,g_2,g_3\}$ into $n-2$ bundles, each of value at least $\mu^n_i(M)$ to $i$. We call such a partition a certificate. Let $P=(P_1, \ldots, P_n)$ be an MMS partition of agent $i$ in the original instance $\calI$. Without loss of generality, assume $\{g_1,g_2,g_3\} \subseteq P_1 \cup \ldots \cup P_j$ for $j \leq 3$. If $\{g_1,g_2,g_3\} \subseteq P_1 \cup P_2$, then $(P_3, \ldots, P_{n-1}, P_n)$ is a certificate. Otherwise, let $g_i \in P_i$ for $i \in [3]$. We have 
    \begin{align*}
        v_i(P_2 \cup P_3 \setminus \{g_2,g_3\}) 
        =
        v_i(P_2) + v_i(P_3) - (v_i(g_2) + v_i(g_3)) 
        \geq 2\mu^n_i(M) - \mu^n_i(M) \geq \mu^n_i(M).
    \end{align*}
    Hence $(P_2 \cup P_3 \setminus \{g_2,g_3\}, P_4, \ldots, P_n)$ is a certificate.
\end{proof}

\lemValidS*
\begin{proof}
    For all agents $i$ who is assigned a bundle $B$ in $S^\alpha$, we have $v_i(B) \geq \alpha \cdot \mu_i(\calI)$. 
    First consider the case that $S^\alpha$ removes two agents $i \neq j$. Since the reduction rule $R^\alpha_1$ is not applicable, for all the remaining agents $a$, we have $v_a(\{n,n+1\}) < \alpha \cdot \mu^n_a(M) \leq \mu^n_a(M)$. Hence, by lemma \ref{lem:new-reduction}, $\mu^{n-2}_a(M \setminus \{1,n,n+1\}) \geq \mu^n_a(M)$. 

    Now consider the case that $S^\alpha$ removes only one agent $i$. Let $a \neq i$ be a remaining agent after allocating $\{1,n+1\}$ to $i$. We have $v_a(\{1,n+1\}) < \alpha \cdot \mu^n_i(M) \leq \mu^n_i(M)$, otherwise, $S^\alpha$ would allocate $\{1,n\}$ to $a$ (or some agent other than $i$). 
    By Lemma \ref{lem:reduce-2goods}, $\mu^{n-1}_a(M \setminus \{1,n+1\}) \geq \mu^n_a(M)$.     
\end{proof}

\thmnOverThree*
\begin{proof}
    We prove for ordered instances $\GenInstance$, $\bagfillRR(\mathtt{T-reduce(\calI,\alpha)},\alpha)$ returns {a simple} $\alpha$-MMS allocation with at most $\floor{n^*/3}$ many distinct copies. Then, from \Cref{prop:ordered_simple}, the theorem follows.

    Let $n$ be the number of agents after applying $\mathtt{T-reduce(\calI,\alpha)}$. Since $T^\alpha$ and $R^\alpha_k$ are valid reductions for all $k < |M|/n^*$, all agents who receive a bag during the $\mathtt{T-reduce(\calI,\alpha)}$ receive at least $\alpha$ fraction of their MMS value and the MMS value of the remaining agents do not decrease after this process. {Without loss of generality, let the MMS value of the remaining agents be $1$ at this point of the algorithm.} All the agents who receive a bag during the round robin phase of $\bagfillRR$ get at least $\alpha$ fraction of their MMS. Let agent $i$ be an agent who receives bag $B_j$ after the round robin phase. Let $g$ be the good that was duplicated and added to $B_j$ right before allocating it to agent $i$. Since $g \leq n'$, $v_i(g) \geq v_i(n')$. By Lemma \ref{lem:duplic-suffice-2}, {$v_i(B_j \cup \{g\}) \geq \alpha$}. Therefore, the final allocation is $\alpha$-MMS.

    We now give an upper bound on the total number of copies and prove that the generated MMS allocation is simple, namely that there are $t \le \floor{n^*/3}$ distinct copies,
    and that the $2t$ instances, composed of these copies and their corresponding original goods, are allocated to $2t$ distinct agents.
    Let $k=n^*-n$ be the number of agents who receive a bag in $\mathtt{T-reduce(\calI,\alpha)}$. Only when applying $T^\alpha$, new duplications are introduced and one duplication per three agents that are removed. So during $\mathtt{T-reduce(\calI,\alpha)}$, at most $\floor{k/3}$ distinct copies are used. Each instance of these good is allocated to a different agent. Furthermore, all the instances of the goods that are allocated during $T^\alpha$ along with their owners, are removed from the instance. Therefore, these goods cannot be copied more than once and their owners cannot get any other goods later on. In the duplication phase of $\bagfillRR$, we duplicate items $1,2, \ldots, n'$ once and allocate the original and the duplications to $2n'$ different agents. By Corollary \ref{cor:halfRemaining}, $n' \leq \floor{n/3}$. Therefore, in total at most $\floor{n^*/3}$ distinct copies are used and the allocation is simple. 
\end{proof}

\subsection{Picking Sequence Fails for Allocation with Copies}\label{sec:picking_equence_fails}
Consider three agent $\{a,b,c\}$ with additive valuations over three goods $x,y,z$. Let the ordinal valuations be as follows:
\begin{eqnarray*}
    v_a :& x \succ z \succ y \\
    v_b :& x \succ z \succ y \\
    v_c :& x \succ y \succ z
\end{eqnarray*}
Consider the ordered counterpart (see \Cref{def:ordered_inst}) of this instance and the allocation $A'$ below, in which every good is copied once.
\begin{eqnarray*}
    A'_a = \{1,2\}, \qquad
    A'_b = \{1,3\}, \qquad
    A'_c = \{2,3\}
\end{eqnarray*}
We interpret this (virtual) allocation as inducing the following value guarantee: for each good $j$ that agent $i$ receives in $A'$, she is entitled to a good whose value is at least her $j$ most-valuable good in the original instance. E.g., agent $b$ in the above must receive $x$ and either of $y$ or $z$.

One can generalize the picking sequence algorithm (\cite{bouveret2016characterizing}, see also \Cref{prop:ordered_reduction_wo_copies}), to a setting with copies in two ways. Assume we have a virtual allocation $A'$, where each good has at most $k$ extra copies. 

One way is to run $k+1$ instances of the picking sequence algorithm, where in each run each good in $M$ is allocated at most once. The first agent to pick is an agent which holds (one of the instances of) the virtual good $1$, the second holds $2$ and so on.
Another way, is to allow all agents which hold good $1$ to pick first, then all agents which hold good $2$, and so on. 
In this interpretation agent $i$ picks the most-valued good $g$ such that $g \notin A_i$, and $g$ was allocated no more than $k+1$ times (so each original good is copied at most $k$ times as well). 
Below we show that neither of these interpretations work if we break the ties arbitrarily. In fact, the second one fails for any tie-breaking rule, see \Cref{example:no_reduction_to_order}.

The failure of the two generalizations follows from the fact that an agent is not allowed to have two instances of the same good. If valuations were such that any $q$ copies of any good $g$ would yield a value of $q\cdot v(g)$, the same proof as in \cite{bouveret2016characterizing} would carry over to this setting.
We show examples in which the two generalizations fail when at most one extra copy is made (i.e., for $k=1$).
For the first generalization, the example below shows that breaking ties arbitrarily may fail.
\begin{example}
Consider the valuations and (virtual) allocation presented at the beginning of the section.
Assume ties are broken such that the picking sequences in the two rounds are $a,c,b$ and $b,a,c$, respectively.
The first run results in $(A_a,A_b,A_c) = (\{x\}, \{z\}, \{y\})$. 
In the second run, since the picking sequence is $b,a,c$, before agent $c$ gets to pick the allocation would be $(A_a,A_b,A_c) = (\{x,z\}, \{z,x\}, \{y\})$ which leaves no item for $c$ to pick in the last round.
\end{example}

For the second generalization, consider the following example, in which the value guarantees do not hold, regardless of tie-breaking.
\begin{example}\label{example:no_reduction_to_order}
\begin{eqnarray*}
    v_a :& x \succ y \succ z \\
    v_b :& z \succ y \succ x \\
    v_c :& z \succ y \succ x
\end{eqnarray*}
with the allocation
\begin{eqnarray*}
    A'_a = \{1,2,3\}, \qquad
    A'_b = \{1\}, \qquad
    A'_c = \{2,3\}
\end{eqnarray*}
\textit{Agents which hold good 1 pick:}
$(A_a,A_b,A_c) = (\{x\}, \{z\}, \emptyset)$. \\
\textit{Agents which hold 2 pick:}
$(A_a,A_b,A_c) = (\{x,y\}, \{z\}, \{z\})$. \\
\textit{Agents which hold 3 pick:}
Agent a cannot pick another good.    
\end{example}

\end{document}